\documentclass[
  aps,
  prl,
  a4paper,
  english,
  reprint,
  superscriptaddress
]{revtex4-1}
\usepackage[T1]{fontenc}
\usepackage[utf8]{inputenc}
\usepackage{graphicx}
\usepackage[american]{babel}
\usepackage{amsmath}
\usepackage{amssymb}
\usepackage{amsthm}
\usepackage[
  caption=false
]{subfig}
\usepackage{refstyle}
\usepackage{algorithm}
\usepackage{listliketab}
\usepackage{algpseudocode}
\usepackage[
  citecolor=blue,
  colorlinks,
  linkcolor=blue,
  urlcolor=blue,
]{hyperref}
\usepackage{xcolor}
\usepackage[babel=true]{csquotes}

\newtheorem{definition}{Definition}
\newtheorem{theorem}{Theorem}
\newtheorem{corollary}{Corollary}[theorem]
\newtheorem{lemma}[theorem]{Lemma}

\begin{document}

\global\long\def\pgr{\mathcal{P}_{\text{gr}}}
\global\long\def\pdb{\mathcal{P}_{\text{db}}}
\global\long\def\pov{\mathcal{P}_{\text{ov}}}
\global\long\def\pn{\mathcal{P}_{0}}
\global\long\def\df{d_{\text{f}}}

\title{Hysteretic percolation from locally optimal individual decisions}

\author{Malte Schr\"oder}
\email{malte.schroeder@tu-dresden.de}
\affiliation{
  Chair for Network Dynamics,
  Center for Advancing Electronics Dresden (cfaed) and Institute for Theoretical Physics, Technical University of Dresden
  01069 Dresden, Germany
}
\affiliation{
  Network Dynamics,
  Max Planck Institute for Dynamics and Self-Organization (MPIDS),
  37077 G\"ottingen, Germany
}

\author{Jan Nagler}
\email{jnagler@ethz.ch}
\affiliation{
   ETH Z\"urich, Wolfgang-Pauli-Strasse 27, CH-8093 Z\"urich, Switzerland
}
\affiliation{
Computational Social Science, Department of Humanities, Social and Political Sciences,
ETH Zurich, Clausiusstrasse 50, CH-8092 Zurich, Switzerland
}
\author{Marc Timme}
\email{marc.timme@tu-dresden.de}
\affiliation{
  Network Dynamics,
  Max Planck Institute for Dynamics and Self-Organization (MPIDS),
  37077 G\"ottingen, Germany
}
\affiliation{
  Chair for Network Dynamics,
  Center for Advancing Electronics Dresden (cfaed) and Institute for Theoretical Physics, Technical University of Dresden
  01069 Dresden, Germany
}

\author{Dirk Witthaut}
\email{d.witthaut@fz-juelich.de}
\affiliation{Forschungszentrum J\"ulich, Institute for Energy and Climate Research - Systems Analysis and Technology Evaluation (IEK-STE), 52428 J\"ulich, Germany }
\affiliation{Institute for Theoretical Physics, University of Cologne, 50937 K\"oln, Germany }

\begin{abstract}
\noindent The emergence of large-scale connectivity underlies the proper functioning of many networked systems, ranging from social networks and technological infrastructure to global trade networks. Percolation theory characterizes network formation following stochastic local rules, while optimization models of network formation assume a single controlling authority or one global objective function. In socio-economic networks, however, network formation is often driven by individual, locally optimal decisions. How such decisions impact connectivity is only poorly understood to date. Here, we study how large-scale connectivity emerges from decisions made by rational agents that individually minimize costs for satisfying their demand. We establish that the solution of the resulting nonlinear optimization model is exactly given by the final state of a local percolation process. This allows us to systematically analyze how locally optimal decisions on the micro-level define the structure of networks on the macroscopic scale. 
\end{abstract}

\maketitle

The proper functioning of networked systems fundamentally relies on their established large-scale connectivity. The global connectivity of social, economic and technological networks, such as the internet, trade and transportation networks, enables global communication and exchange, but also the rapid spreading of diseases \cite{albert00_attack_tolerance, satorras01_epidemic, hufnagel04_forecast, brockmann06_travel, brockmann13_spreading, albert02_network_review, newman03_network_review}. 
The loss of connectivity, or even of a single connection, may cause catastrophic effects such as the collapse of ecological networks, blackouts of power grids and other infrastructures, or even a global economic crisis \cite{sole01_ecological, schweitzer09_economic, buldyrev10_interdependent, havlin12_challenges, elliott14_financial, witthaut16_critical, ronellenfitsch17_power, nagler11_single_links}. Understanding how global connectivity emerges thus constitutes a key challenge in the field of network science.

Two major theoretical approaches have been established for revealing core properties of the emergence of large-scale connectivity. First, the theory of percolation provides fundamental insights about network formation processes by assuming that new links are established stochastically according to some local rule \cite{stauffer_92_percolation_book,grimmett99_percolation_book}.
For such percolation models, a variety of distinct structure-forming phenomena have been observed, where diverse network topologies emerge even for simple link formation rules \cite{newman03_network_review, achlioptas09_explosive, riordan11_explosive, schroder13_crackling, schroeder16_supercritical, dsouza15_review, watts98_smallworld, verma16_emergence, dsouza15_review}. Second, global optimization models explain network formation controlled by a 
central authority or driven by a single global objective function. These models have been studied to construct and understand aspects of the structure of various man-made or biological networks \cite{bertsekas98_optimization, gastner06_distribution, bohn07_optimaltransport, katifori10_optimal, katifori16_optimal, chklovskii02_volumeoptimized, memmesheimer06_designingNetworks, memmesheimer06_designingNeural}.

The formation of many socio-economic networks, however, is driven by local agents making individual decisions based on optimizing their own goals. 
Such settings result in networks constrained by many individual, yet interacting optimization problems. A similar motivation underlies game-theoretic models of network formation \cite{jackson96_strategic, watts01_dynamic, jackson02_network_evolution, jackson08_network_book, easley10_networks_crowds_markets, bala00_noncooperative, konig12_efficiency, even07_smallWorldGameTheory, atabati15_strategic}. These models allow a more detailed analysis of the formation process and the stability of the resulting network. Unsurprisingly, however, they are often hard, if not impossible, to solve, especially for larger networks, which limits mechanistic insights.

In this Letter we study network formation processes based on rational agents that individually optimize their own local objective function. 
Given costs for production and transaction (including transport) in an underlying transport network, each agent satisfies its own demand at minimal costs \cite{krugman91_geography,krugman91b_geography}. 
We establish an exact mapping between the solutions of the resulting nonlinear optimization problems and the states of a local percolation process. This enables us to systematically investigate optimal collective network formation and to reveal discontinuous transitions of the network structure and hysteresis. These effects are independent of the network topology or specific choices of the cost functions. The proposed framework thus bridges (deterministic) network optimization and stochastic local percolation.\\

\textit{From optimization to percolation} --- 
We analyze a network formation model based on the following fundamental network supply problem. 
Consider an underlying network of $N$ nodes and $M$ links, describing agents and \emph{potential} transportation routes, where each agent must satisfy its demand. We study the network of trades that actually evolves between the nodes, similar to bond percolation on an underlying network or random graph \cite{stauffer_92_percolation_book}.

\begin{figure}
\centering
\includegraphics[width=0.5\textwidth]{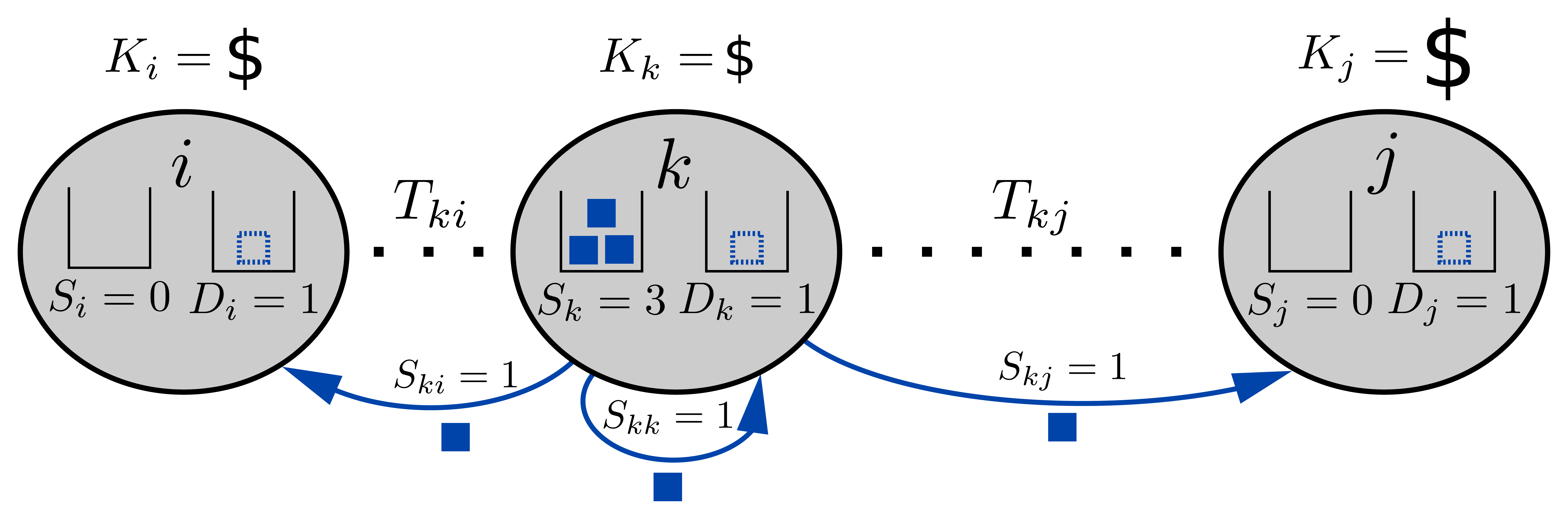}
\caption{
\textbf{Network Supply Problem.} Each agent $i$ satisfies its demand $D_i$ by purchasing supplies $S_{ki}$ from nodes $k$ at minimum costs via the available transport routes (dashed lines). The costs $K_{i} = \sum \nolimits_{k} K^P_{ki} + K^T_{ki}$ for these purchases include production costs $K^P_{ki}$ at node $k$ and transaction costs $K^T_{ki}$ between the nodes [Eq.~(\ref{eq:opt})]. The production costs depend on the total production $S_k = \sum_i S_{ki}$ at node $k$ [Eq.~(\ref{eq:prod_cost})]. The transaction costs are proportional to the effective distance $T_{ki}$ of the transport route between the nodes [Eq.~(\ref{eq:trans_cost})].
}
\label{fig:model_intro}
\end{figure}

Specifically, we assume each node \mbox{$i \in \left\{1,\ldots,N\right\}$} is an agent with a fixed demand $D_i$. The agent satisfies this demand by purchasing supplies $S_{ki} \ge 0$ from any nodes $k$, including itself, under the constraint that $\sum_k S_{ki} = D_i$. (Throughout the manuscript sums run over all nodes, here \mbox{$k \in \left\{1,\ldots,N\right\}$}, unless otherwise noted.) Each agent tries to achieve this with minimal cost
\begin{equation}
	K_i = \sum \nolimits_{k} K^P_{ki} + K^T_{ki} \,,	\label{eq:opt}
\end{equation}
including both production costs $K^P_{ki}$ at node $k$ as well as transaction costs $K^T_{ki}$ between the two nodes (see Fig.~\ref{fig:model_intro}).

The production costs depend nonlinearly on the purchases $S_{ki}$
\begin{equation}
	K^P_{ki} = p_k(S_k) S_{ki} \, \label{eq:prod_cost}
\end{equation}
since the costs per unit $p_k(S_k)$ typically depend on the total production $S_k = \sum \nolimits_j S_{kj}$ at node $k$. Often, the production costs per unit are decreasing, $\mathrm{d} p_k / \mathrm{d}S_k < 0$, accounting for, e.g., increased efficiency with increased production, commonly referred to as economies of scale.

The transaction costs are proportional to the amount of transported goods $S_{ki}$ and to the effective distance $T_{ki}$ between the two nodes,
\begin{equation}
	K^T_{ki} = p_T S_{ki} T_{ki} \,, 	\label{eq:trans_cost}
\end{equation}
where $T_{ki} = \sum_e t_e$ is given as the sum of the distances $t_e$ of all edges $e$ along the (shortest) path between $k$ and $i$ in the underlying transport network. The factor $p_T$ denotes the transaction costs per unit good and unit distance, describing effects of fuel costs or delivery times.

All agents solve their individual nonlinear optimization problem [Eq.(\ref{eq:opt})] simultaneously, defining the network of optimal purchases $S_{ki}$. The resulting state of this network then corresponds to a Nash-equilibrium \cite{osborne94_gameTheory_book}, where no agent can reduce its cost by changing its purchases given that all other purchases remain constant.\\

\begin{figure}
\centering
\includegraphics[width=0.5\textwidth]{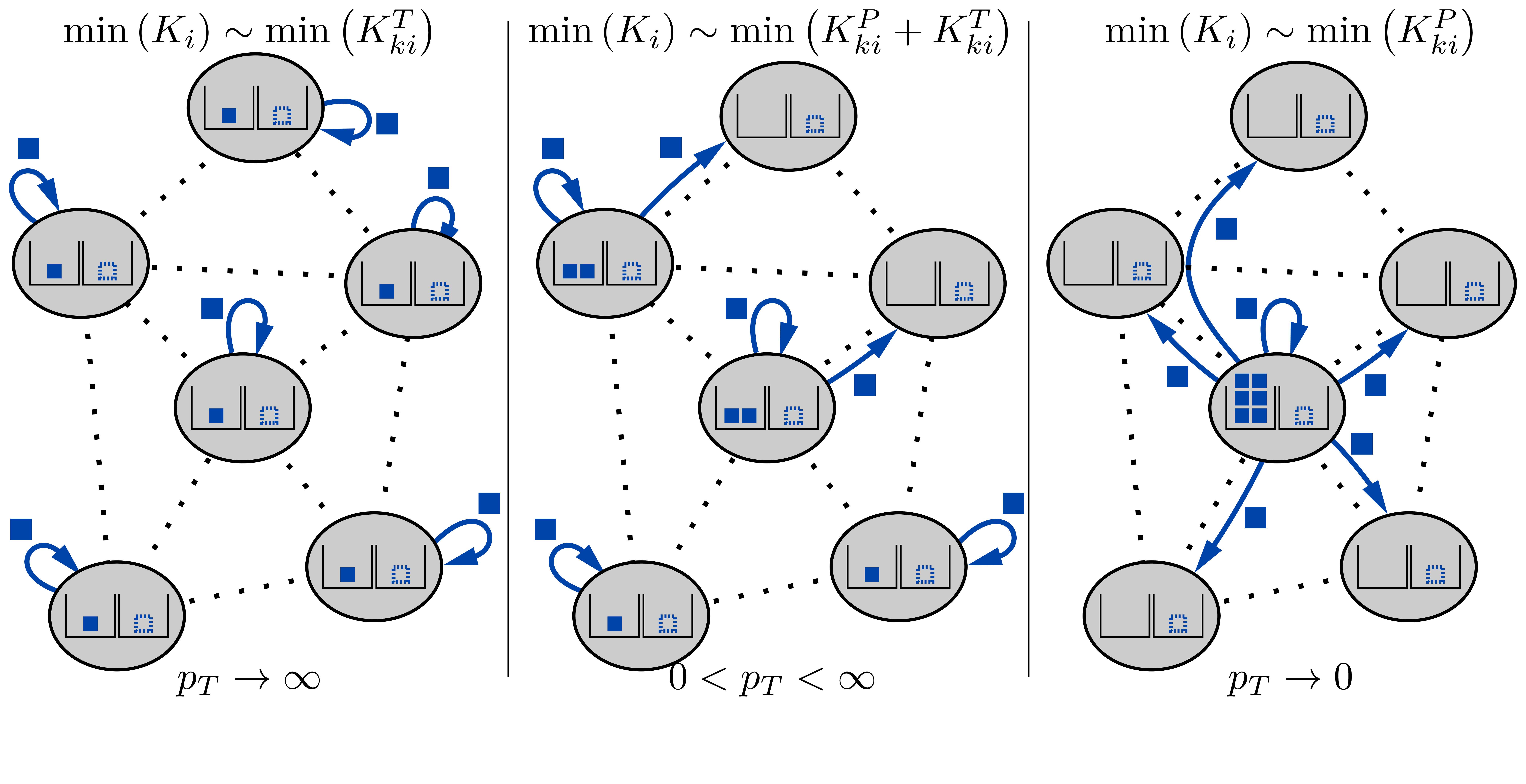}
\caption{
\textbf{Local percolation induced by optimization.} Every agent chooses a single optimal supplier to satisfy its demand if the production costs per unit $p_k(S_k)$ decrease with increasing production (economies of scale). When the transaction costs per unit $p_T$ are large, all agents make only local purchases from their own node. As $p_T$ decreases, the transaction costs decrease and agents start to purchase from other nodes. Some potential transportation routes (dashed lines) become active trade links (arrows). This trade network grows along the shortest paths in the network as transaction costs become smaller. As the transaction costs disappear ($p_T \rightarrow 0$), the network will become globally connected. All agents then share the same supplier, minimizing their production costs. If the demand of all agents is identical $D_i \equiv D$, the solution to the supply problem can be found locally and is given by the final state of a local percolation model.
}
\label{fig:model_schematic}
\end{figure}

\textit{Results} --- 
A simple, yet efficient solution to this problem can be found for non-increasing production costs per unit $p_k$. 
In this case, we find that any agent $i$ chooses a single supplier $i^*$, such that $S_{i^* i} = D_i$ and $S_{ki} = 0$ for $k \neq i^*$. In general each agent would have to check each node in the network to find its optimal supplier. Interestingly, if the demand of all agents is identical, $D_i \equiv D$, this optimal supplier can be found \emph{locally}: An agent $i$ just has to query its direct neighbors about their current suppliers to find its optimal supplier $i^*$. 

Here, we provide a brief intuitive argument: Any purchase of agent $i$ has to be transported via one of its neighbors $j$. Since transaction costs are additive over the transport links and all agents have identical demand, agent $j$ effectively solves the same optimization problem as agent $i$ (minus the transaction costs from $j$ to $i$). When $j$ finds its optimal supplier, this supplier is also a potential optimal supplier of agent $i$ when transporting via $j$. Thus, agent $i$ simply compares the suppliers of all its neighbors (all potentially optimal suppliers, one for each possible path of transport). One of these must then be the optimal supplier for agent $i$ (see Supplemental Material Sec.~I and II for a rigorous proof and details of the simulation \cite{supplement}). 

\begin{figure*}
\centering
\includegraphics[width=0.85\textwidth]{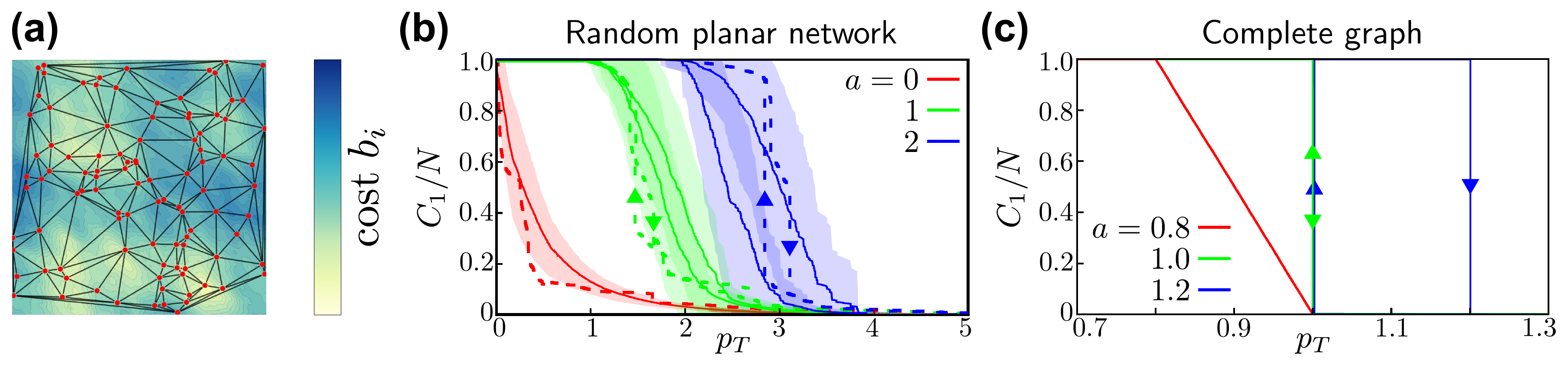}
\caption{
\textbf{From continuous transitions to discontinuous transitions and hysteresis.} The type of the transition to a single globally connected component changes depending on the strength of the economies of scale in the cost function. (a) Example of a small random network embedded in the unit square. The parameters $b_i \in [0,1]$ are given by the real part of a smooth random function $b(x,y)$ generated from $2^{10} \times 2^{10}$ discrete Fourier modes with mean amplitude $0$ and variance $S(\omega_x, \omega_y) = (\omega_x^2 + \omega_y^2)^{-2}$ (see also Supplemental Material Sec.~III \cite{supplement}). (b) Single realizations (dashed lines) of the evolution of the relative size of the largest cluster for a random planar network with $N=10^4$ nodes and average and standard deviation over 100 random realizations of the $b_i$ (solid lines and shading). (c) The predicted behavior for a completely connected network in the mean field limit $N \rightarrow \infty$ (see text). Weak economies of scale ($a < a_c$) lead to a continuous growth of the largest cluster. Sufficiently strong economies of scale ($a > a_c$) lead to a discontinuous transition. Reversing the process, i.e., increasing the transaction costs, leads to a direct reversal for weak economies of scale, but hysteresis is observed for strong economies of scale.
}
\label{fig:network_transition}
\end{figure*}

We investigate this local percolation model starting with large transaction costs, $p_T = \infty$, and, correspondingly, only internal production $i^* = i$ and $S_{ii} = D_i$. As $p_T$ decreases, the transaction costs decrease and agents minimize their total costs by establishing external purchases $S_{ki}$ from nodes with lower production costs. Finally, transaction costs disappear at $p_T = 0$ and all agents will have the same supplier minimizing their production costs (Fig.~\ref{fig:model_schematic}). 
We study the size $C(i^*)$ of the connected components (clusters) in the network defined by these purchases, this means the number of agents $\left\{i_1,i_2,\dots\right\}$ with the same supplier $i^*$. As for standard percolation we record the size $C_1(p_T)$ of the currently largest cluster (and $C_2$ for the second largest cluster and so on).
In the following examples we consider linearly decreasing production costs per unit $p_k(S_k) = b_k - a S_k$, where $a \ge 0$ directly quantifies the strength of the economies of scale. The results are qualitatively unchanged for all forms of decreasing $p_k$ (see Supplemental Material \cite{supplement}).\\

\textit{Discontinuous percolation and hysteresis} --- 
We illustrate the emergence of connectivity in a random spatially embedded network in  Fig.~\ref{fig:network_transition}(a,b), revealing the importance of economies of scale. Weak economies of scale (small $a$) lead to a continuous growth of the largest cluster. Sufficiently strong economies of scale lead to a discontinuous evolution of the size of the largest cluster in the network. A microscopic decrease of the transaction costs triggers a cascade of decisions: As the cluster size increases, the production costs of its supplier decrease and a large fraction of agents join this connected component. In the language of percolation, a giant connected component emerges in a continuous (weak economies of scale) or discontinuous (strong economies of scale) phase transition. 

Moreover, multiple stable states exist for sufficiently strong economies of scale. In an intermediate interval of transaction costs $p_T$ the network settles on one of the 
possible structures, depending on the previous state of the network: hysteresis emerges. 
Thus, a large cluster may remain stable after it has emerged for decreasing $p_T$, even when $p_T$ is increased again [Fig. \ref{fig:network_transition}(b,c)].\\

\textit{Underlying mechanism} --- 
To understand the mechanism underlying these different transitions, we analyze a mathematically tractable system in detail. We consider a network of all-to-all coupled units with demand $D_i = 1/N$ separated by effective distances $t_e = 1$ for all edges $e$. We take  $b_i = i/N$ for $i \in \left\{1,\dots,N\right\}$, approximating the uniform distribution $b_i \in \left[0,1\right]$ in the limit of large system size $N$.

We now track individual decisions by considering the cost per unit $K_i(k)$ agent $i$ pays for purchases at node $k$. Since transaction costs across all links are identical, the first link to be established will be between the node with the highest production cost (node $N$) and the one with the smallest 
(node $1$). This happens when the cost per unit $K_N(1)$ for agent $N$ to import from node $1$ become smaller than the cost $K_N(N)$ to buy internally: $K_N(1) = 1/N - 2 a/N + p_T < 1 - a / N = K_N(N)$, that is for $p_T < p_T^N = 1-1/N+a/N$. Similarly, we can calculate when the next link between agent $N-1$ and node $1$ is established: $K_{N-1}(1) = 1/N - 3 a/N + p_T < (N-1)/N - a / N = K_{N-1}(N-1)$, that is for $p_T < p_T^{N-1} = 1-2/N+2a/N$. The other agents follow the same pattern.

Considering the two links, we now have to distinguish two cases: if $a < a_c = 1$, then $p_T^{N-1} < p_T^{N}$ and the agents $N$ and $N-1$ will establish their links sequentially at different values of $p_T$. The largest cluster will grow continuously with a slope of $(1/N) / (p_T^{N-1} - p_T^N) = 1/(a-1)$. However, if the economies of scale are stronger ($a \ge a_c = 1$), the cost at node $1$ decrease sufficiently for the next link to be established immediately since $p_T^{N-1} \ge p_T^N$. The cluster grows discontinuously in a single cascade. 
If the economies of scale are even stronger ($a > a_c = 1$), the cluster is stable with respect to single agents changing their supplier for larger values of $p_T$, causing hysteresis when increasing $p_T$. This qualitative behavior is independent of the network topology in the sense that for sufficiently large economies of scale the transition will always become discontinuous (see Supplemetal Material \cite{supplement}).\\

\textit{Impact of network topology} --- 
Besides changes in production costs, the growth of the trade network is also determined by changes in transaction costs, that means by the underlying physical transportation network. 
For the same economies of scale we find different routes of network formation depending on the structure of the network. If the network diameter (the longest shortest path between any two nodes) is small, the paths in the network are short and only one cluster emerges. If the diameter is large, multiple large clusters appear. This difference is already evident when comparing the spatially embedded (large diameter) and complete network (small diameter) for $a = 1$ [compare Fig.~\ref{fig:network_transition}(b) and \ref{fig:network_transition}(c)].

\begin{figure}
\centering
\includegraphics[width=0.5\textwidth]{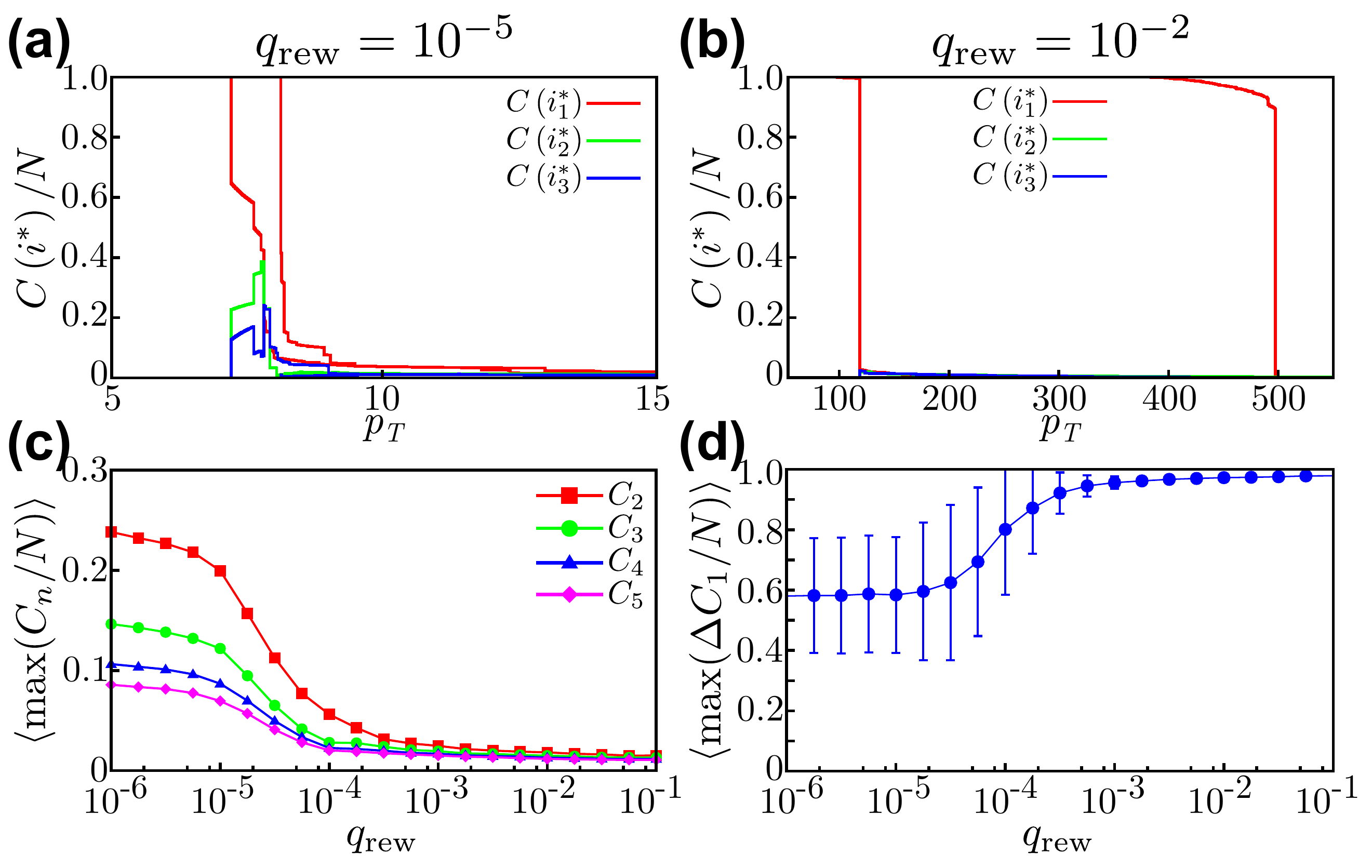}
\caption{\textbf{Impact of the network structure on the percolation transition.}
Size and discontinuities of large connected components for an underlying Watts-Strogatz small world network ($N = 10^4$, $k=8$, see Supplemental Material Sec.~III and V \cite{supplement}). (a),(b)~Single realizations for the size of the components $C(i^*)$ of three specific large suppliers. In a network with large diameter [panel (a), $q_\mathrm{rew} = 10^{-5}$] multiple clusters grow simultaneously and merge for small $p_T$.
In a network with small diameter [panel (b), $q_\mathrm{rew} = 10^{-2}$] one large cluster emerges in a single cascade.
(c)~Maximum size of the $n$-th largest cluster as a function of the topological randomness $q_\mathrm{rew}$ (error bars omitted for visibility). (d)~Largest change of the size of the largest cluster (error bars indicate the standard deviation). Averages are taken over $100$ realizations.
}
\label{fig:small_world}
\end{figure}

\begin{figure*}
\centering
\includegraphics[width=0.9\textwidth]{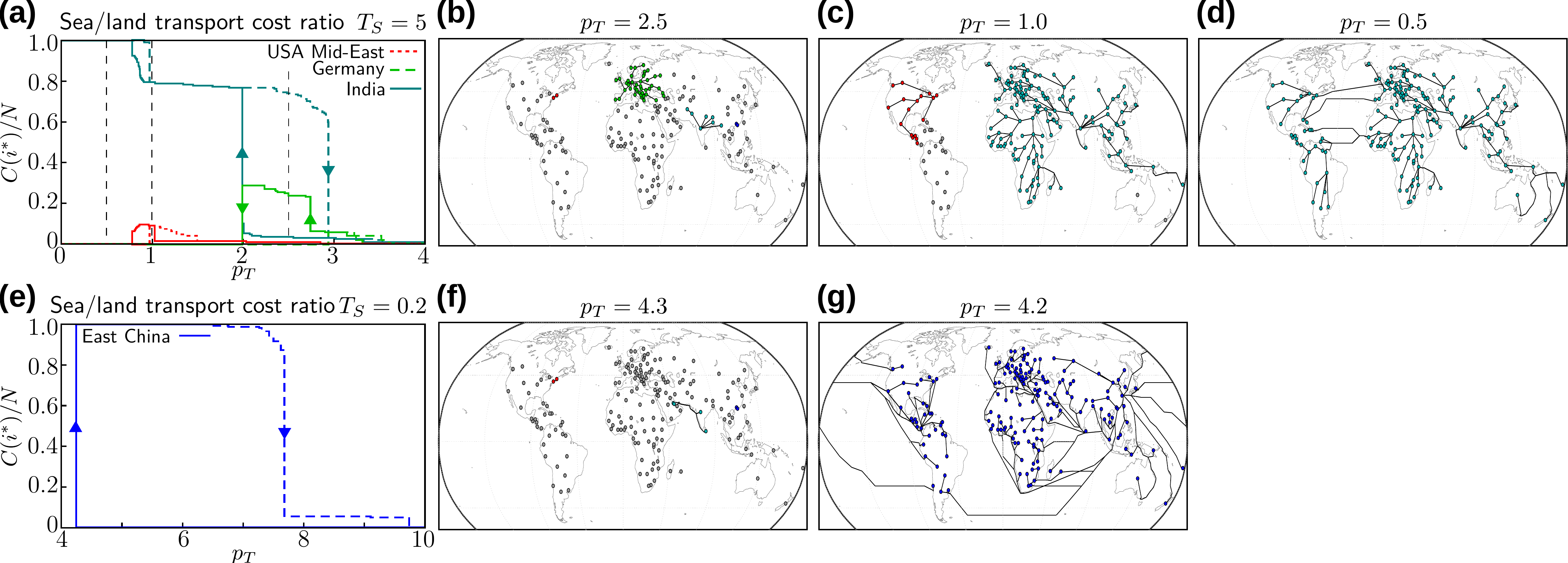}
\caption{\textbf{Preferred modes of transport change network evolution.} 
Global connectivity induced by locally optimal decisions in a model of a world transport network (see Supplemental Material Sec.~VII for details \cite{supplement}). (a)~Evolution of the size $C(i^*)$ of specific clusters identified by a supplier $i^*$ when transport via sea is more expensive than transport via land (by a factor $T_S=5$). Due to the high costs of sea transport, the network diameter is large and multiple large clusters can emerge in different parts of the world.  (b)-(d)~Network structure and active trade links for different values of the transactions costs $p_T$. Land routes are preferred to transport via sea.
(e)~Evolution of the size of the emerging cluster for small sea transport costs ($T_S=0.2$). Due to sea routes connecting most countries cheaply, the network diameter is small and only one large component emerges in a single large cascade. This transition happens for larger $p_T$ as the overall diameter of the network is much smaller. (f),(g)~The state of the network immediately before and after the transition.}
\label{fig:worldmap}
\end{figure*}

To systematically study this effect we fix the economies of scale $a=1$ and consider a network class introduced by Watts and Strogatz \cite{watts98_smallworld}. Starting from a regular ring network with a large diameter where each node is connected to its $k$ neighbors, each link is randomly rewired with probability $q_\mathrm{rew}$. This introduces shortcuts and reduces the diameter of the network. If the diameter is large [$q_\mathrm{rew}$ small, Fig.~\ref{fig:small_world}\,(a)], different suppliers can attract large clusters of agents from 
their local part of the network when $p_T$ decreases. Fig.~\ref{fig:small_world}\,(c) illustrates the maximum size of the $n$-th largest cluster, showing that multiple large clusters emerge for $q_\mathrm{rew} \le 10^{-5}$ (less than one shortcut per node). Only for small values of $p_T$ do these clusters interact and finally merge in small cascades to a single giant cluster. If the diameter is small [$q_\mathrm{rew}$ large, Fig.~\ref{fig:small_world}\,(b)], only a single cluster emerges, attracting nodes from all parts of the network. The largest cluster then grows in a single cascade until it fills the entire network [Fig.~\ref{fig:small_world}\,(c),(d)].

To illustrate this effect with a realistic network topology, we consider an elementary model of a world transport network (Fig.~\ref{fig:worldmap}). The nodes of the network represent individual countries or regions and links represent transport routes via land between neighboring countries or shipping routes via sea (for a detailed description of the model parameters and all data see Supplemental Material Sec.~VII \cite{supplement}). Similarly to the small-world network, we explore different network structures by varying the costs for different modes of transportation, modifying the effective distances of transport via land and sea. If transport via sea is expensive the network has a large diameter, similar to the random planar network [Fig.~\ref{fig:network_transition}(a,b)]. Multiple large clusters appear in different regions of the world, merging when $p_T$ becomes small [Fig.~\ref{fig:worldmap}\,(a)-(d)]. Conversely, if transport via sea is cheap, the network becomes densely connected with a small diameter, similar to the complete graph [Fig.~\ref{fig:network_transition}(c)], and a single largest cluster grows in a sudden cascade [Fig.~\ref{fig:worldmap}\,(e)-(g)].\\

\textit{Discussion} --- 
Taken together, we have proposed a class of network formation models that demonstrates how fundamental aspects of local economic decisions impact global network formation. In contrast to random link addition extensively studied before \cite{stauffer_92_percolation_book, dsouza15_review, bollobas98_random}, link addition in this model is driven by individual decisions: each node minimizes its own costs to satisfy a fixed demand by establishing a trade network across a given network of (potential) transport routes. 
The model class is general in the sense that the illustrated phenomena are independent of both the details of the network topology as well as of the details of the cost functions. Specifically, whereas we have analyzed the model with linearly decreasing production costs per unit, linearity is not required. Any decreasing cost function yields qualitatively the same results (see Supplemental Material \cite{supplement}).

In an analytically solvable limit we showed that the solution of the resulting interacting optimization problems exactly maps to the final state of a local percolation process, thereby enabling a systematic analysis of the transition from localized trading to a macroscopically connected trade network. Importantly, the network evolution exhibits hysteresis, that means decreasing an external factor (e.g. the transaction cost per unit) may induce global interactions but increasing the same factor does not directly reverse this system-scale impact. 

These results directly link deterministic, optimization-based approaches of network formation \cite{gastner06_distribution, bohn07_optimaltransport, katifori10_optimal, katifori16_optimal, chklovskii02_volumeoptimized, memmesheimer06_designingNeural, jackson96_strategic, watts01_dynamic, jackson02_network_evolution, bala00_noncooperative, konig12_efficiency, even07_smallWorldGameTheory, atabati15_strategic} and prototypical percolation processes \cite{newman03_network_review, achlioptas09_explosive, riordan11_explosive, schroder13_crackling, schroeder16_supercritical, dsouza15_review, watts98_smallworld, verma16_emergence, dsouza15_review} based on purely random link addition. Specifically, the model illustrates the connection of (discontinuous) percolation to geographically distributed trade networks \cite{krugman91_geography, krugman91b_geography}. More generally, our framework shows how economic factors and the actual connectivity may shape the structure of socio-economic networks through individual locally optimal decisions.\\

We thank J. Többen  and S. Klipp for helpful discussions.
We gratefully acknowledge support from the G\"ottingen Graduate School for Neurosciences and Molecular Biosciences (DFG Grant GSC 226/2), from the Helmholtz association (grant no. VH-NG-1025), the German Ministry for Education and Research (BMBF grants no. 03SF0472A-F), the German Science Foundation (DFG) by a grant towards the Cluster of Excellence \emph{Center for Advancing Electronics Dresden} (cfaed), the ETH Risk Center (RC SP 08-15) and SNF Grant {\em The Anatomy of systemic financial risk}, No. 162776.\\

%

\bibliography{manuscript_individualDecisions}

\onecolumngrid

\newpage
\null
\newpage

\setcounter{figure}{0}
\renewcommand{\thefigure}{S\arabic{figure}}

{\centering \Large  \textbf{Appendix}\\
accompanying the manuscript\\
\textbf{Hysteretic percolation from locally optimal individual decisions\\}}
{\center \normalsize{Malte Schr\"oder, Jan Nagler, Marc Timme and Dirk Witthaut}\flushright}
\quad\\
\quad\\

In the main manuscript we introduced a network formation model where link addition is based on individual decisions of local agents in a fundamental network supply problem. In general the resulting individual optimization problems are too complex to be solved efficiently for large systems. We discussed how these optimization problems map to a local percolation model with an efficient solution, bridging the gap between global optimization models of network formation and stochastic local percolation models. This mapping allowed us to efficiently study the phenomena emerging in this network formation model.

In this supplemental material we give the rigorous proofs for the mapping as well as additional details and examples for the model. First, we consider the mapping of the optimization problem to a percolation model. In particular, we discuss in detail the assumptions and requirements on the cost functions and other parameters for this mapping to be valid and give a rigorous proof. We also describe the simulation procedure resulting from these conditions. Second, we discuss in more detail the parameter choices made for the examples shown in the main manuscript in relation to properties of the model, describing how and why we take the limit of an increasing number of nodes $N \rightarrow \infty$ as a fine-graining rather than an expansion of the system in relation to the underlying network topology (for example random and small world networks). Third, we provide an additional example for an analytically solvable system in a continuous interpretation of the model, linking the equilibrium states and transitions to bifurcations of stable states of a corresponding self-consistency equation. Finally, we include a detailed description of the data used for the world transport network simulations in the main manuscript and a brief comparison of the exact solution for the world transport network to the approximate solution with our local percolation algorithm, which is shown in the main manuscript.\\

\section*{The network supply problem}

To begin, we reiterate the basic idea of the network formation model as defined in the main manuscript:
Starting with an underlying network of $N$ nodes and $M$ \emph{potential} transport links, we consider each node $i$ to be an agent with a fixed demand $D_i$ that it satisfies by purchases $S_{ki}$ from nodes $k$, possibly internally from $k=i$. These purchases incur a cost of two parts: First, the production cost $K^P_{ki} = S_{ki} p_k(S_k)$ at node $k$ depending on the total amount of production $S_k$ at that node. Second, a transaction cost (e.g., transport cost) in the underlying network, $K^T_{ki} = p_T S_{ki} T_{ki}$, where $T_{ki} = \sum_{e \in \Pi(k,i)} t_e$ is the sum over the transaction costs of all links $e$ along the shortest (cheapest) path $\Pi(k,i)$ from $k$ to $i$. The parameter $p_T$ describes the transaction costs per unit across the network, i.e., the importance of transaction costs relative to production costs. We assume that each agent only decides its purchases and the production at a node is always given as the sum of all purchases made from that node, $S_k = \sum_i S_{ki}$. As in the main manuscript, sums run over all nodes $i \in \left\{1,\dots,N\right\}$ unless noted otherwise. Similarly, each agent always satisfies its demand exactly, such that $D_i = \sum_k S_{ki}$. Each agent $i$ then individually optimizes its purchases $S_{ki}$ to minimize its cost $K_i$ while satisfying its demand:
\begin{align}
 \mathrm{minimize} \quad & \quad K_i = \sum_k S_{ki} p_k(S_k) + p_T \sum_k \left[ S_{ki} \sum_{e \in \Pi(k,i)} t_e \right] \label{eq:restrictions}\\
 \mathrm{subject\;to} \quad & \quad \sum_k S_{ki} = D_i	\nonumber\\
 		\quad & \quad S_{ki} \ge 0 \;\forall\; k \nonumber\,.
\end{align}

We consider the evolution of the network of trades described by the purchases $S_{ki}$ when the transaction cost $p_T$ decrease, starting from $p_T = \infty$ with initially only internal production $S_{ii} = D_i$ and $S_{ki} = 0$ for $i \neq k$.
Large transaction costs per unit $p_T$ mean that no external purchases are made across the network. As $p_T$ decreases, some agents start buying at cheaper neighbors and small, localized connected components (clusters) start to grow. Specifically, as in standard percolation theory, we are interested in the evolution of the size of the largest cluster $C_1(p_T)$.

To allow easier readability, we use different indices to denote nodes or agents based on their role in the current context wherever possible. We summarize the variables together with their meaning in the following table:\\

\begin{listliketab}
    \storestyleof{itemize} 
    \begin{tabular}{ll}
     $N$ & The number of nodes in the system, system size \\
     $e_{i j}$ & Link (potential transport link) between nodes $i$ and $j$ \\
     $t_e$, $t_{ij}$ & Distance for the link $e$ from node $i$ to $j$\\
     $\Pi(k,i) \quad$ & Shortest (cheapest) path from $k$ to $i$, $\Pi(k,i) = \left\{e_{k j_1}, e_{j_1 j_2}, \dots, e_{j_n i}\right\}$\\
     $T_{ki}$ & Total distance between the nodes $k$ and $i$, $T_{ki} = \sum_{e \in \Pi(k,i)} t_e$\\
     $p_T$ & Transaction cost per unit good and distance, relative importance of transaction cost\\
     \hline
	 $i,j$ & Index of an agent (currently looking to make purchases)\\
	 $k,l$ & Index of a node (currently considered as a supplier)\\
	 $i',j'$ & Index of the current supplier of agent $i,j$\\
	 $i^*,j^*$ & Index of the/an optimal supplier of agent $i,j$\\
	 $D_i$ & The (fixed) demand of agent $i$\\
	 $S_k$ & The current, total production of node $k$\\
	 $S_k^i$ & The production of node $k$ disregarding possible purchases by agent $i$, $S_k^i = S_k - S_{ki}$\\
	 $S_{ki}$ & The amount agent $i$ purchases from node $k$\\
	 $p_k\left(S_k\right)$ \quad & The production cost per unit at node $k$ with a total production $S_k$\\
	 \hline
	 $K_i$ & The total cost, including transaction costs and production costs, of all purchases of agent $i$\\
	 $K_i(k)$ & The total cost of agent $i$ when \emph{only} purchasing from node $k$\\
    \end{tabular} 
\end{listliketab}
\\

\newpage

\section{Mapping to a local percolation model}

In general the optimization problem described above is very complex and quickly grows intractable for larger systems. Here we derive a general updating scheme in terms of a local percolation model, showing that it exactly solves all individual optimization problems under certain conditions. We first consider the cost function of a single agent $i$
\begin{equation}
	K_i = \sum_k \left[ S_{ki} p_k\left(S_k^i + S_{ki}\right) + p_T S_{ki} T_{ki} \right], \label{eqn:cost}
\end{equation}
where $S_k^i = \sum_{j \neq i} S_{kj}$ denotes the production of node $k$ ignoring purchases by agent $i$, such that $S_k = S_k^i + S_{ki}$, and we write the total distance over all links from $k$ to $i$ as $T_{ki}$. The problem becomes considerably easier for the family of cost functions describing so called economies of scale, this means decreasing production costs per unit with increasing production. In this case the problem of finding the optimal $S_{ki}$ reduces to finding a single optimal supplier.

\begin{lemma}
Given an individual minimization problem of agent $i$ defined by Eq.~(\ref{eq:restrictions}) with non-increasing production costs per unit at each node, $\frac{\partial p_k}{\partial S_k}\left(S_k\right) \le 0$ for all $k \in \left\{1, \dots, N\right\}$, there is a node $i^*$ such that $K_i$ is minimal with $S_{i^*i} = D_i$ and $S_{ki} = 0$ for $k \neq i^*$.
\end{lemma}

\begin{proof}
Choose the node $l$ for which
\begin{equation}
  p_l^{\rm eff}(S_l^i + D_i) = p_l(S_l^i + D_i) + p_T T_{l i}
\end{equation}
is smallest. Together with $\partial p_k/\partial S_k \le 0$ for all
nodes $k$ we then obtain 
\begin{equation}
   p_l^{\rm eff}(S_l^i + D_i) \le p_k^{\rm eff}(S_k^i + D_i) \le p_k^{\rm eff}(S_k^i + S_{ki})
\end{equation}
for all nodes $k \in \left\{1,\dots,N\right\}$ and arbitrary purchases $0 \le S_{ki} \le D_i$.
Using the constraint that $D_i = \sum_k S_{ki}$ this implies
\begin{align}
   & K_i(S_{1i}=0,S_{2i}=0,\ldots,S_{li}=D_i,S_{l+1,i}=0,\ldots) \nonumber \\
   & \quad = D_i p_l^{\rm eff}(S_l^i + D_i)  \nonumber \\
   & \quad \le \sum_k S_{ki} p_k^{\rm eff}(S_k^i + S_{ki}) \nonumber \\
   & \quad = K_i(S_{1i},S_{2i},S_{3i},\ldots) 
\end{align}
for all possible purchases $(S_{1i},S_{2i},S_{3i},\ldots)$.
Thus the costs $K_i$ assume a global minimum if agent $i$ satisfies its entire demand by
purchases from a single node $i^* = l$, this means for
\begin{equation}
  S_{i^*i} = D_i \quad \mbox{and} \quad S_{ki} = 0 \; \mbox{for} \; k \neq i^* \,. \nonumber
\end{equation}
\end{proof}

In particular this class of non-decreasing functions includes the affine-linear production costs per uni $p_k\left(S_k\right) = b_k - a \cdot S_k$ with $a \ge 0$ used in the main manuscript. However, we are not restricted to identical slopes or even functions of the same form for different nodes. We therefore cover a broad range of cases where the cost functions of all nodes are non-increasing in the range of possible production $S_k \in \left[0 , \sum_i D_i \right]$.\\

With the above simplification, we now consider the optimization problem of agent $i$ of finding the best $i^*$ to minimize
\begin{align}
 K_i\left(i^*\right)	&= D_i p_{i^*}\left(S_{i^*}^i + D_i\right) + p_T D_i \sum_{e \in \Pi(i^*,i)} t_e \,.
\end{align}
To find an equilibrium, we can in principle simply try all possible alternatives, as formulated in the following algorithm. However, such a brute-force approach can become quickly infeasible for large networks such that we will consider further simplifications in the following sections.
 
\begin{definition}[Equilibrium]
Consider the optimization problem defined above [Eq.~(\ref{eq:restrictions})] with non-increasing production costs per unit $p_k(S_k)$ for all agents $k \in \left\{ 1,2,\dots,N \right\}$. An agent $i$ is in equilibrium at its current supplier $i'$, if and only if there is no node $k$ such that $K_i(k) < K_i(i')$. This means the agent cannot change its supplier to reduce its costs, given that all other agents keep their current supplier.

We say that the network is in equilibrium if all agents in the network are in equilibrium.
\end{definition} 

Note that this corresponds to the notion of a Nash-equilibrium, where no agent can decrease its cost by changing only its own supplier. Note also, while the network might be in equilibrium, this does not necessarily mean that the total cost for all agents, $\sum_i K_i$, is minimal.
 
\newpage
\subsection*{General algorithm}
Consider the optimization problem defined above with \emph{non-increasing production costs per unit $p_k(S_k)$} for all agents $k = \left\{1,2,\dots,N\right\}$. Given a fixed value $p_T$, for each agent $i$, individually,
\begin{align}
 \mathrm{find}	\quad & \quad i^* \in \left\{1,\dots,N\right\} \nonumber\\
 \mathrm{minimizing} \quad & \quad K_i\left(i^*\right) = D_i p_{i^*}\left(S_{i^*}^i + D_i\right) + p_T D_i \sum_{e \in \Pi(i^*,i)} t_e \,.
\end{align}
In general, the following algorithm yields an equilibrium (though not necessarily the global optimum) when it terminates:

\begin{algorithm}[H]
\caption{General optimization}
\label{alg:alg1}
\begin{algorithmic}[1]
\Repeat
	\State $\mathcal{U} \leftarrow \varnothing$
	\ForAll{$i \in \left\{1,2,\ldots,N\right\}$}
		\ForAll{$k \in \left\{1,2,\ldots,N\right\} $}
			\State Calculate $K_i(k)$ [where $S_{ki} = D_i$ and $ S_{li} = 0$ for all $l \neq k$].
			\If{ $K_i(k) < K_i(i')$ where currently $S_{i'i} = D_i$ }
				\State $\mathcal{U} \leftarrow \mathcal{U} \cup \left\{K_i(k)\right\}$
			\EndIf
		\EndFor
	\EndFor
	\If{$\mathcal{U} \neq \varnothing$}
		\State Calculate $\left\{i,i^*\right\}$ such that $K_i(i^*) = \min_{\mathcal{U}} K_i(k)$
		\State $S_{i^* i} \leftarrow D_i$ 
		\State $S_{i' i} \leftarrow 0$
	\EndIf
\Until{$\mathcal{U} = \varnothing$}
\end{algorithmic}
\end{algorithm}

In words, for every step we calculate the cost for all agents and possible suppliers. If an agent $i$ would prefer a new supplier $k$ to its current supplier $i'$ we add the corresponding cost to the list of possible updates $\mathcal{U}$. Here, we use $K_i(k)$ to denote both the final cost for agent $i$ as well as the corresponding update itself. In the end we execute the update with the smallest final cost (clearly any update will find the currently optimal choice $i^*$ for agent $i$, since we check all possible suppliers in the network). We repeat this process until no further update is possible.

\begin{theorem}
	Algorithm \ref{alg:alg1} will always terminate in finite time. When it terminates the network will be in an equilibrium state, where no agent can reduce its costs by changing its supplier.
\end{theorem}

\begin{proof}
Firstly, the algorithm only terminates when there are no possible updates, i.e., when no single agent can further reduce its costs by changing its supplier. As we assumed non-increasing production costs per unit, a single supplier is always an optimal solution. Thus no single agent $i$ can reduce its costs at all and we always end in an equilibrium.

Conversely, if the network is in equilibrium, no agent can reduce its cost. The update list will be empty and the algorithm will terminate.\\

Secondly, consider the possibility of the algorithm not terminating at all. This can only happen if the algorithm runs into a loop, executing the same updates repeating a finite set of states of the network. We now show that such a loop is impossible:

Assume an initial state $\mathcal{I}$ of the network, defined by the pairs of agents and current suppliers $\mathcal{I} = \left\{\left(i,i'\right), \dots\right\}$. A loop can only occur if the network now reaches a state $\mathcal{I}-1$, where a single agent $i$ has a different supplier $k$, but would like to return to the supplier $i'$ it has in state $\mathcal{I}$ ($K_i(i') < K_i(k)$).

To reach this state we need a chain of updates including, at least, one update $u_a$ where agent $i$ switches supplier to $k$. However, if the reverse update $u_a^{-1}$ is strictly beneficial in state $\mathcal{I}-1$, clearly $u_a$ is not beneficial in state $\mathcal{I}$. Thus the chain needs to consist of more than a single update.\\

Given that one update alone is not possible, we need another update to enable it. We thus have to consider a chain of updates $u_b \circ u_a \circ u_b^{-1}$, where $u_b$ enables $u_a$ and is reversed in the end to arrive at state $\mathcal{I}-1$. The main idea is illustrated in Fig.~\ref{fig:proof_termination}. There are only two possible ways how a beneficial update $u_b$ can enable $u_a$: 

(1) either an agent $j$ switches from its supplier $j'$ to the the new supplier $k$. But then the update $u_a$ decreases the cost $K_j(k)$ further. The reverse update $u_b^{-1}$ would then never be executed since $K_j(k) < K_j(j')$.

(2) $u_b$ could have an agent $j$ switch away from the current supplier $j' = i'$ to a different node $l$. But then $u_a$ further increases the cost $K_k(i')$ and we again find $K_j(l) < K_j(i')$, making the reverse update impossible.\\

Following this line of argument, we need another update $u_c$ to enable the update $u_b^{-1}$, e.g., a chain of updates $u_b \circ u_a \circ u_c \circ u_b^{-1} \circ u_c^{-1}$. This clearly requires a partial ordering of the updates, such that $u_b$, $u_c$, $u_b^{-1}$ and  $u_c^{-1}$ are executed in this order (otherwise the update $u_c$ will have no effect on the update $u_b^{-1}$). However, by the same logic as above $u_b^{-1}$ would change the cost unfavorably for $u_c^{-1}$. We then find the same problem that $u_c^{-1}$ would never be executed at the end of the chain. Repeating the argument, we find the same problem for every finite chain of updates.\\

Thus, we cannot construct a finite chain of updates to reach $\mathcal{I}-1$. Consequently, we cannot repeat a state and thus would visit every possible state, but then the algorithm would terminate in (one of) the equilibrium states. Together, we find that the algorithm must terminate in finite time.

\end{proof}

\begin{figure}[h]
\centering
\includegraphics[width=0.75\textwidth]{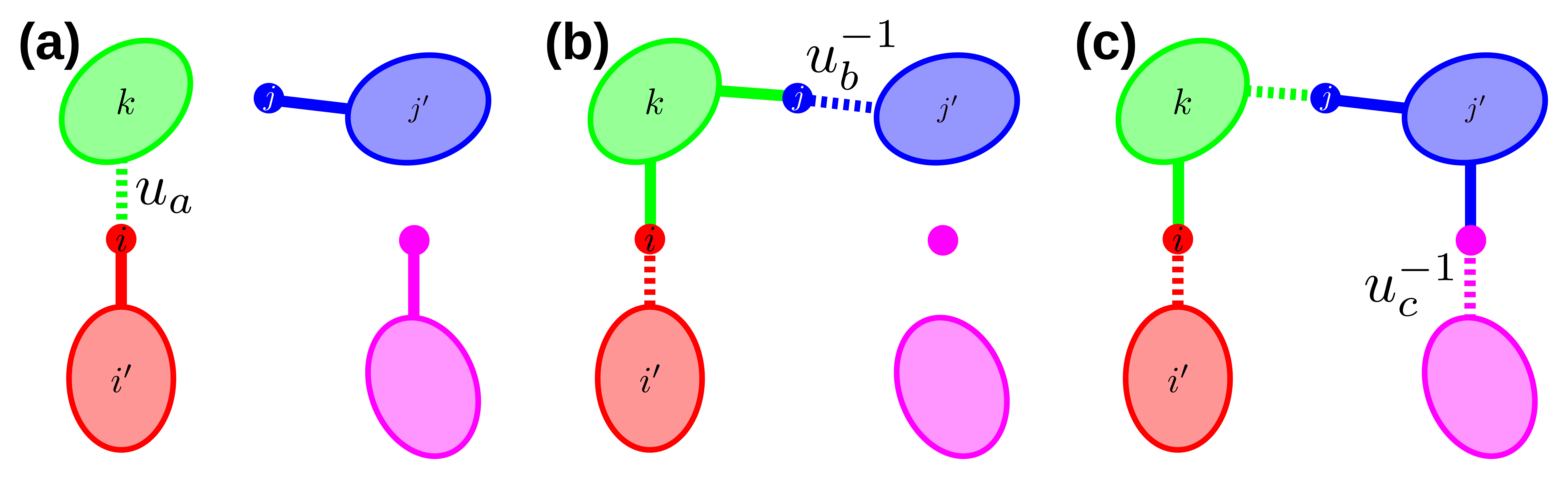} \\
\caption{Illustration of the main argument. Starting from an initial state $\mathcal{I}$, we cannot find a chain of updates that leaves us in a state $\mathcal{I} - 1$ where agent $i$ buys from $k$ but wants to switch back to $i'$. (a) Clearly, a single update is not enough, as $i$ would not want to switch to $k$. (b) Using another agent $j$ to enable the switch of $i$ to $k$ leaves $j$ buying at $k$, the reverse update required to reach state $\mathcal{I} - 1$ is not possible. (c) Following the same logic, a chain with three (or more) updates will also not be able to reach state $\mathcal{I} - 1$ (see text).}
\label{fig:proof_termination}
\end{figure}

\newpage
\subsection{When local information is sufficient}

With the above simplification of non-increasing production costs per unit, we still consider the individual optimization problem of node $i$ of finding the best $i^*$ to minimize
\begin{align}
 K_i(i^*) &= D_i p_{i^*}(S_{i^*}^i + D_i) + p_T D_i \sum_{e \in \Pi(i^*,i)} t_e \,.
\end{align}
In principal, we consider every node in the network as a possible supplier. We now show that in certain cases we can further simplify this optimization by only considering the local neighborhood of agent $i$ to find its best supplier $i^*$. 

\begin{lemma}
Given the optimization problem defined above with non-increasing production costs per unit $p_k(S_k)$ and \emph{homogeneous demand $D_i = D$} for all agents $i \in \left\{1,\dots,N\right\}$. If all other agents are in equilibrium, then either $i$ is its own optimal supplier or there exists a neighbor $j$ of $i$ such that the optimal supplier $j^*$ of $j$ is also an optimal supplier of $i$.
\end{lemma}

\begin{proof}
Obviously, if $i^* = i$, we are done.

Otherwise, consider the external optimal supplier $i^*$ of $i$. The path from $i^*$ to $i$ will pass through a neighbor $j$ of $i$, such that
\begin{align}
   \Pi(i^*,i) &= (e_{i^* k_1}, e_{k_1 k_2}, \ldots, e_{j, i}) \\
   T_{i^*,i} &= t_{i^* k_1} + t_{k_1 k_2} + \ldots + t_{j i}\,. 
\end{align} 
We give a proof by contradiction, assuming that $j^*$ is not an optimal supplier of $i$, i.e., $j^* \neq i^*$ and $ K_i(j^*) > K_i(i^*)$.
Using the fact that $j^*$ is the optimal supplier of $j$ we obtain the inequality 
\begin{eqnarray}
	\frac{K_j(j^*)}{D} &=& p_{j^*}\left(S_{j^*}^j + D\right) + p_T  \sum_{e \in \Pi(j^*,j)} t_e \\
						 &\le& p_{i^*}\left(S_{i^*}^{i} + D\right) + p_T  \sum_{e \in \Pi(i^*,j)} t_e = \frac{K_j(i^*)}{D} \,.
\end{eqnarray}

Since $j$ is in equilibrium we have 
\begin{eqnarray}
	S_{j^*}^j + D &=& S_{j^*}	\nonumber\\
	S_{i^*}^j + D &=& S_{i^*} + D	\,,
\end{eqnarray}
and for agent $i$ we have also 
\begin{eqnarray}
	S_{j^*} \le S_{j^*}^i + D &\le& S_{j^*} + D \nonumber\\
	S_{i^*} \le S_{i^*}^i + D &\le& S_{i^*} + D\,,
\end{eqnarray}
where $S_{j^*}^j$ again denotes the production of $j^*$ ignoring purchases of agent $j$ and the second set of inequalities simply states that $i$ may currently be buying from any node, possibly even $i^*$ or $j^*$.
Using these observations together with the fact that all $p_{k}$ are non-increasing we find
\begin{eqnarray}
&& - p_{i^*}\left(S_{i^*}^{i} + D\right) + p_{i^*}\left(S_{i^*}^j + D\right) - p_{j^*}\left(S_{j^*}^j + D\right) + p_{j^*}\left(S_{j^*}^{i} + D\right) \nonumber\\
&=&\underbrace{- p_{i^*}\left(S_{i^*}^{i} + D\right) + p_{i^*}\left(S_{i^*} + D\right)}_{\le 0} \underbrace{- p_{j^*}\left(S_{j^*}\right) + p_{j^*}\left(S_{j^*}^{i} + D\right)}_{\le 0} \nonumber\\
&\le& 0	\, \label{eq:demand_condition}.
\end{eqnarray}

We can combine the above inequalities to obtain
\begin{align}
	\frac{K_i(j^*)}{D} & = p_{j^*}\left(S_{j^*}^{i} + D\right) + p_T  \sum_{e \in \Pi(j^*,i)} t_e \nonumber\\
	&\le p_{j^*}\left(S_{j^*}^{i} + D\right) + p_T  \sum_{e \in \Pi(j^*,j)} t_e + p_T t_{ji}\nonumber\\
	&= p_{j^*}\left(S_{j^*}^j + D\right) + p_T  \sum_{e \in \Pi(j^*,j)} t_e - p_{j^*}\left(S_{j^*}^j + D\right) + p_{j^*}\left(S_{j^*}^{i} + D\right) + p_T t_{ji} \nonumber\\
	&\le p_{i^*}\left(S_{i^*}^j + D\right) + p_T  \sum_{e \in \Pi(i^*,j)} t_e - p_{j^*}\left(S_{j^*}^j + D\right) + p_{j^*}\left(S_{j^*}^{i} + D\right) + p_T t_{ji} \nonumber\\
	&= p_{i^*}\left(S_{i^*}^{i} + D\right) + p_T  \sum_{e \in \Pi(i^*,j)} t_e + p_T t_{ji} \nonumber\\
	&
	\quad\quad - p_{i^*}\left(S_{i^*}^{i} + D\right) + p_{i^*}\left(S_{i^*}^j + D\right) - p_{j^*}\left(S_{j^*}^j + D\right) + p_{j^*}\left(S_{j^*}^{i} + D\right)  \nonumber \\
	&\le p_{i^*}\left(S_{i^*}^{i} + D\right) + p_T  \sum_{e \in \Pi(i^*,i)} t_e \nonumber \\
	& = \frac{K_i(i^*)}{D}\,,
	\label{eq:ineq_Ki_Kj}
\end{align}
where the first inequality is due to $j^*$ being an optimal supplier for $j$ and the second inequality follows from Eq.~(\ref{eq:demand_condition}).

Hence, we know that
\begin{equation}
    K_i(j^*) \le K_i(i^*) ,
\end{equation}  
contradicting our original assumption. Therefore $j^*$ has to be an optimal supplier of $i$ which concludes the proof.
\end{proof}
\quad\\

We note that this optimal supplier $i^* = j^*$ does not necessarily have to be unique, it is possible that there are multiple optimal suppliers with identical (minimal) costs.

We briefly summarize the results so far: if we have non-increasing production costs per unit and sufficiently homogeneous demand [Eq.~(\ref{eq:demand_condition})], we can in principle solve the optimization problem of a single agent $i$ by considering only its local, direct neighborhood. In particular, this condition is fulfilled under the stronger assumption of identical demand $D_i = D_j = D$ for all agents. However, we assumed a network in (almost) equilibrium, specifically we assumed that $j$ is already buying at $j^*$. If this is not the case, it remains to be shown that we can order simultaneous updates in such a way that we perform only local updates and still always find the optimal supplier for each agent.\\

\begin{corollary}
We call the set $\left\{i_1,i_2,\dots\right\}$ of agents with $i^*$ as their optimal (and current) supplier the \emph{cluster} $C\left(i^*\right)$ of $i^*$. In a network in equilibrium the individual clusters are connected: for every node $i \in C\left(i^*\right)$ there is a shortest (cheapest) path \mbox{$\Pi(i^*,i) = \left(i^*, j_1, j_2, \ldots, j_n, i\right)$} from $i^*$ to $i$ such that all nodes $j_1,j_2 \dots j_n \in C\left(i^*\right)$. The active links in the clusters follow the shortest path tree from the supplier.
\end{corollary}

\newpage
\subsection{Local percolation algorithm}

Above we showed that we can always find an optimal supplier for an agent $i$, given the other agents are in equilibrium, by considering its neighbors when we have non-increasing production costs per unit and homogeneous demand $D_i = D$ for all agents $i \in \left\{ 1,\ldots,N \right\}$. We use this result to derive a simplified, local percolation algorithm to solve this optimization problem in an efficient way. Since individual purchase decisions can lead to larger changes in the network, we still need to consider updates in a network out of equilibrium. Here, one can easily construct examples where there exists an agent that cannot find its (currently) optimal supplier locally. In these cases, however, other agents will also want to update their purchases and at least one agent can find its optimal supplier locally. Below, we introduce the algorithm and the ordering of the local updates such that we always update agents who can find their optimal supplier locally. We show that this defines a \emph{local} percolation rule taking the network from \emph{any} state into an equilibrium with only local updates to find optimal suppliers.

\begin{algorithm}[H]
\caption{Local percolation algorithm}
\label{alg:alg2}
\begin{algorithmic}[1]
\Repeat
	\State $\mathcal{U} \leftarrow \varnothing$. 
	\ForAll{$i \in \left\{ 1,\ldots,N \right\}$}
		\ForAll{$j \in \mathrm{Neighborhood}(i)$}
        	\State Calculate $K_i(j')$ and $K_i(i')$, where currently, $S_{i'i} = D_i$ and $S_{j'j} = D_j$
    		\If {$ K_i(j') < K_i(i')$ where currently $S_{i'i} = D_i$ }
    			\State $\mathcal{U} \leftarrow \mathcal{U} \cup \left\{ K_i(j') \right\}$
    		\EndIf 
        \EndFor
        \State Calculate $K_i(i)$
    		\If {$ K_i(i) < K_i(i') $}
    			\State $\mathcal{U} \leftarrow \mathcal{U} \cup \left\{ K_i(i) \right\}$
    		\EndIf 
    \EndFor
	\If{$\mathcal{U} \neq \varnothing$}
		\State Calculate $\left\{i,i^*\right\}$ such that $K_i(i^*) = \min_{\mathcal{U}} K_i(k)$
		\State $S_{i^* i} \leftarrow D_i$ 
		\State $S_{i' i} \leftarrow 0$
	\EndIf
\Until{$\mathcal{U} = \varnothing$}   
\end{algorithmic}
\end{algorithm}

This algorithm is a local version of the general algorithm, where we check for every agent $i$ only the suppliers $j'$ of its neighbors $j$ and itself as possible new suppliers. As before we then execute the update with the smallest final cost and repeat the process until no further update is possible.\\

\begin{theorem}
Consider the optimization problem defined above with non-increasing production cost per unit $p_k(S_k)$ and homogeneous demand $D_i = D$ for all agents $i \in \left\{1,\ldots,N\right\}$. Algorithm 2 only executes updates finding the optimal supplier for an agent and terminates in finite time in an equilibrium state, where no agent can reduce its costs further.\
\end{theorem}
\quad\\

We now proof this theorem with two supporting Lemmas:

\begin{lemma}
Given algorithm \ref{alg:alg2}, for every $K_i(k) \in \mathcal{U}$, either $k = i^*$ or there is $K_j(l) \in \mathcal{U}$ with $K_j(l) < K_i(k)$.
\end{lemma}

\begin{proof}
If $k = i^*$ we are done. \\

Otherwise, assuming that $k \neq i^*$, then also $i' \neq i^*$ since $K_i(i^*) < K_i(k) < K_i(i')$, thus $i$ is not already buying at $i^*$. Then, following the path $\Pi(i^*,i)$ backwards from $i$ to $i^*$, we will find an agent $j$ (possibly $j = i^*$ or $j = i$), where a local update $K_j(i^*)$ is possible and the current supplier $j' \neq i^*$.

We now consider this update $K_j(i^*)$. First we show that this update is in the update list by contradiction. Assume $K_j(j') \le K_j(i^*)$, then similar to the above proof we find
\begin{eqnarray}
	\frac{K_i(j')}{D} &=& p_{j'}\left(S_{j'}^i + D\right) + p_T T_{j' i} \nonumber\\
			&\le& p_{j'}\left(S_{j'}^i + D\right) + p_T T_{j' j} + p_T T_{ji} \nonumber\\
			&\le& p_{j'}\left(S_{j'}\right) + p_T T_{j' j} + p_T T_{ji} \nonumber\\
			&=& K_j(j') + p_T T_{ji} \nonumber\\
			&\le& K_j(i^*) + p_T T_{ji} \nonumber\\
			&\le& p_{i^*}\left(S_{i^*} + D\right) + p_T T_{i^*j} + p_T T_{ji} \nonumber\\
			&=& \frac{K_i(i^*)}{D} \,,
\end{eqnarray}
leading to the contradiction $K_i(j') < K_i(i^*)$ ($i^*$ is not the optimal supplier). Otherwise we have $K_i(j') = K_i(i^*)$ and we can simple repeat the above argument with $i^* \leftarrow j'$ (at most until we reach $j = i$).
This shows that there is an update $K_j(i^*)$ on the list $\mathcal{U}$ for some optimal supplier $i^*$ of $i$ and an agent $j$ on the path $\Pi(i^*,i)$.\\
Finally, we find
\begin{eqnarray}
	\frac{K_j(i^*)}{D} 	&=& p_{i^*}\left(S_{i^*}^j + D\right) + p_T T_{i^* j}	\nonumber\\
						&\le& p_{i^*}\left(S_{i^*}^i + D\right) + p_T T_{i^* j} + p_T T_{j i} \nonumber\\
						&=& p_{i^*}\left(S_{i^*}^i + D\right) + p_T T_{i^* i} \nonumber\\
						&=& p_{i^*}\left(S_{i^*}^i + D\right) + p_T T_{i^* i} \nonumber\\
						&=& \frac{K_i(i^*)}{D} \,,
\end{eqnarray}
since $S_{i^*}^j = S_{i^*}^i = S_{i^*}$ since none of the two agents are buying at $i^*$. This means that the update $K_j(i^*)$ with $K_j(i^*) \le K_i(i^*) < K_i(k)$ is in the update list before $K_i(k)$, which concludes the proof.
\end{proof}
\quad\\

It then follows directly that:
\begin{corollary}
The update on the list $K_i(k) = \min_\mathcal{U} K_j(l)$ is optimal with $k=i^*$.
\end{corollary}
\quad\\

\begin{lemma}
The network is in an equilibrium, if and only if $\mathcal{U} = \varnothing$.
\end{lemma}

\begin{proof}
We proof this by showing that the network is not in equilibrium, if and only if $\mathcal{U} \neq \varnothing$.\\

Clearly, if there is a possible update on the list, the network is not in equilibrium.\\

Otherwise, if the network is not in equilibrium there is an agent $i$ currently buying from $i' \neq i^*$ which is not its optimal supplier. We then have $K_i(i^*) < K_i(i')$. If there was no update on the list, then also the update $K_i(i^*)$ via $i$'s neighbor $j$ is not on the list. However, this can only be because $j$ is buying from $j' \neq i^*$. But then either
\begin{itemize}
	\item[(1)] the node $j'$ would also be an optimal supplier for $i$ (following the same chain of inequalities as above) and the update $K_i(j')$ would be on the list 
	\item[(2)] or $j$ would also want to update to buy at $i^*$.
\end{itemize}
In the second case we simply repeat the argument above with $i \leftarrow j$ (at most until $i \leftarrow i^*$ itself) until we find a possible update which must be on the list. This shows that there is always an update on the list when the network is not in equilibrium and thus concludes the proof.
\end{proof}
\quad\\

\begin{proof}[Proof of Theorem 3]
By means of Corollary 4.1 we know that the update executed is always an optimum one. Hence every iteration of the algorithm will find the optimal supplier for an agent.
Lemma 5 shows that, when the algorithm terminates as soon as $\mathcal{U}$ is empty, the network is in equilibrium. The proof that the algorithm always terminates in finite time is identical to the one given above for Algorithm 1.
Together these two results proof theorem 3.
\end{proof}

In contrast to $\mathcal{O}(N^2)$ updates that are checked in the general algorithm, we further reduce the number of possibilities to $\mathcal{O}(M)$ local updates, where $M$ is the number of links in the network. In sparse networks, where $\mathcal{O}(M) = \mathcal{O}(N)$, networks this significantly reduces the required operations of the algorithm.\\

Together, we have now shown that for non-increasing production costs per unit $p_k(S_k)$ and identical demand of all agents $D_i = D_j = D$, we can \emph{exactly} map the interacting individual optimization problems to a local percolation model when ordering simultaneous updates by their new final cost. These updates will always lead to an equilibrium state where no agent wants to change suppliers and agents never make sub-optimal updates. Note, however, that due to the individual optimization the resulting equilibrium is, in general, not the globally optimal state with minimal total cost $\sum_i K_i$ over all agents.\\

Here, we always use the algorithm in the following way: We start in an equilibrium state, usually the trivial equilibrium $i^* = i$ for all $i$ at $p_T = \infty$. Then we consider a decrease (or increase) of $p_T$ and calculate the next value $p_T^{ji}$ for which a agent would change its supplier (i.e., when $\mathcal{U} \neq \varnothing$). We then execute the update and iterate until no further updates occur at the current value of $p_T$. As we have shown that the algorithm always makes correct updates, it is also possible to study other initial conditions, reactions to changes in the network structure or sudden large changes of $p_T$.\\

Finally, in the last part of this section we briefly give small examples illustrating some assumptions in the proofs above, such as homogeneous demand, before describing the actual implementation in the next section.

\newpage

\subsection{Examples}

In order to illustrate the model, the proofs given above and especially the assumptions going into the proofs, we give two examples of simple systems. For all examples we will use the same cost per unit functions as in the main manuscript $p_k\left(S_k\right) = b_k - a \cdot S_k$ with $a=1$.

\begin{figure}[h]
\centering
\includegraphics[width=0.27\textwidth]{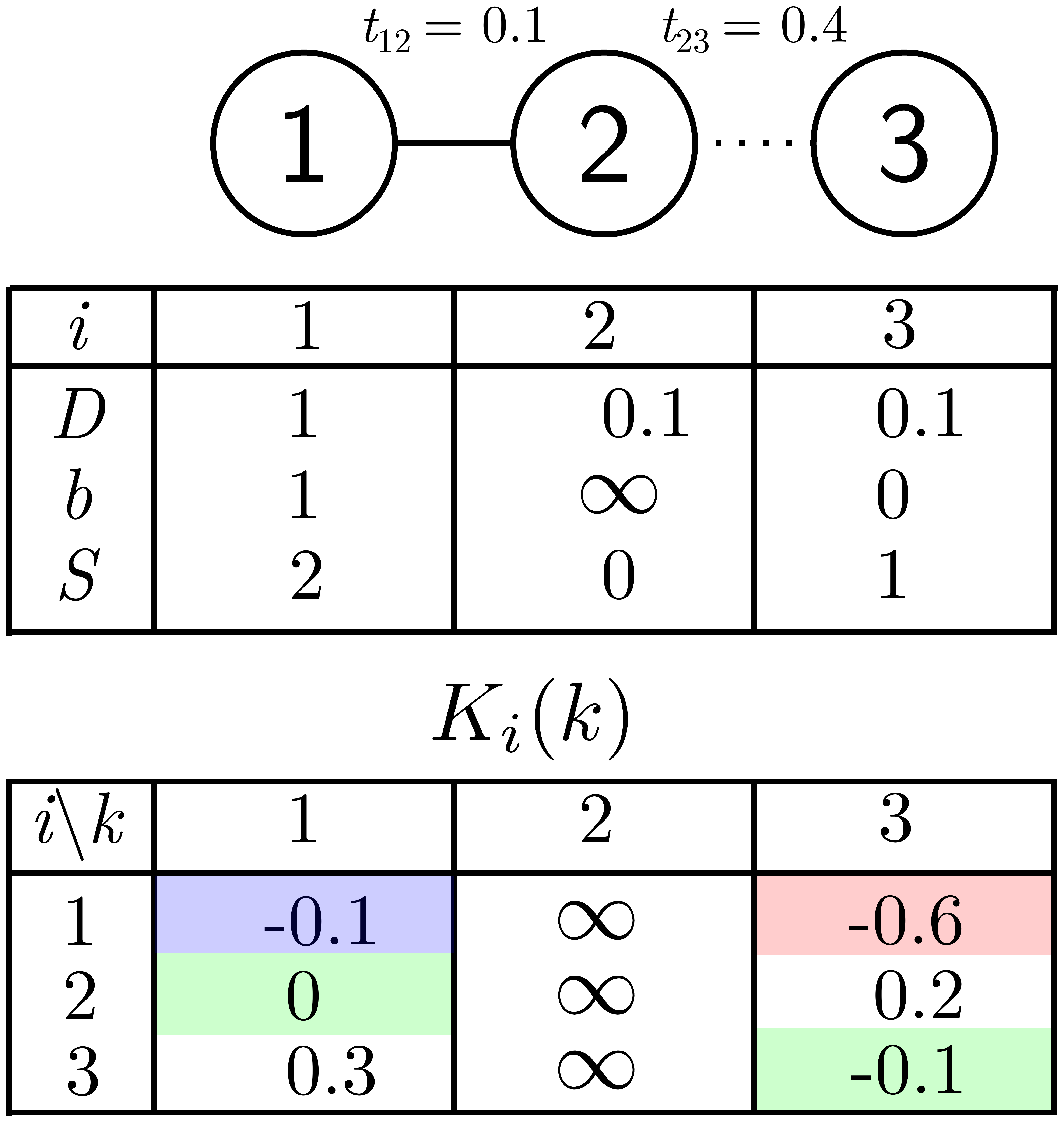} \\
\label{fig:example_homogeneous_demand}
\caption{Example for non-homogeneous demand at $p_T = 1$. The tables describe the model parameters and the costs if agent $i$ would buy (is buying) at node $k$. Current suppliers are marked in green, if they are optimal, otherwise they are marked in blue. Optimal decisions are marked in yellow or red, depending on whether they can/cannot be found by a local update, respectively. Here, agent $1$ could buy optimally from node $3$ but cannot find out about it via a local update (marked in red), since agent $2$ is buying from $1$. Conversely, agent $2$ and $3$ are already buying optimally at node $1$ and $3$, respectively (marked in green). Therefore the only possible update is $1$ buying from $3$ but it cannot be found with a local update. This illustrates the requirement of homogeneous demand for the local percolation algorithm.}
\end{figure}

\begin{figure}[h]
\centering
\includegraphics[width=0.45\textwidth]{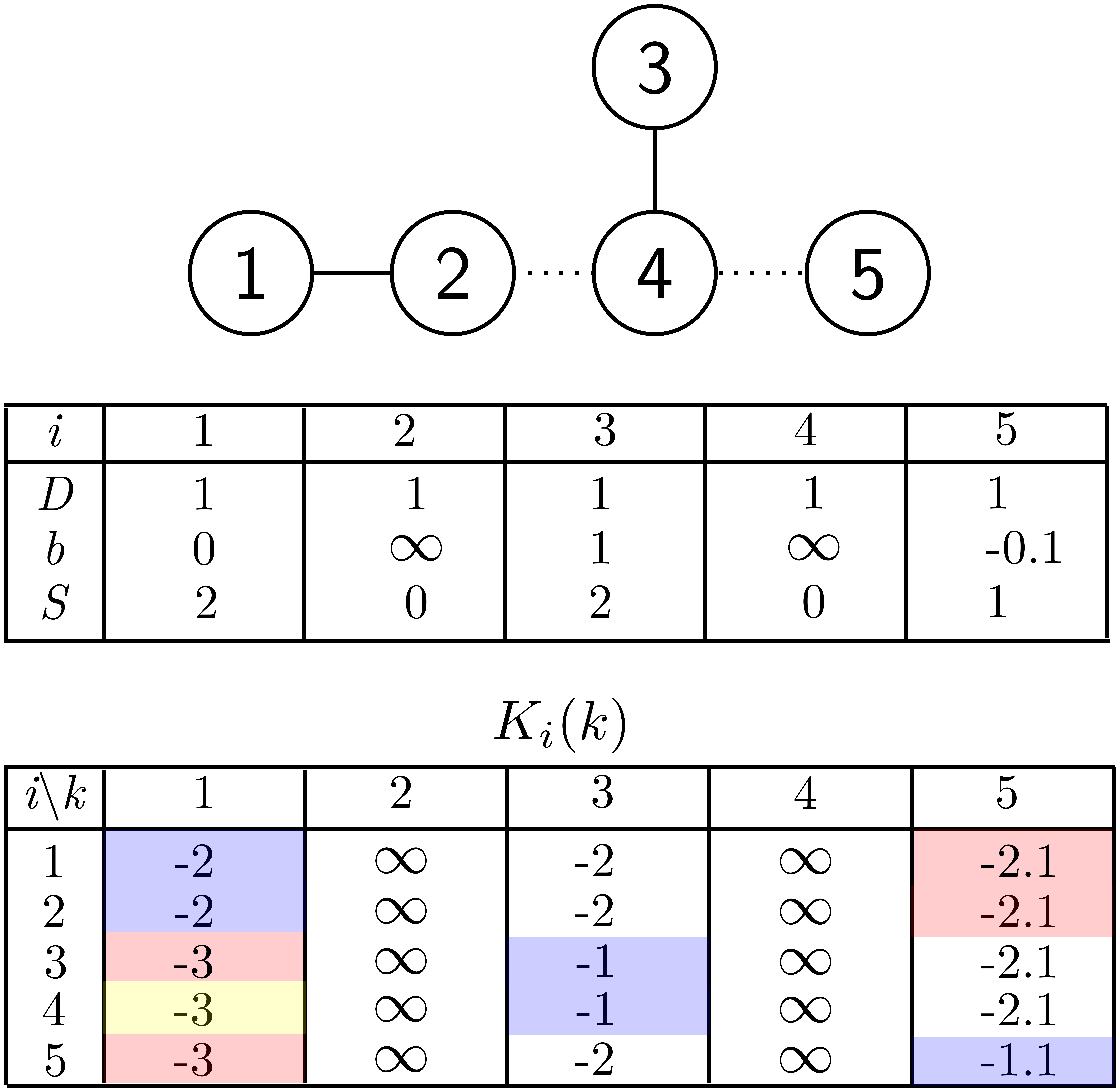} \\
\label{fig:example_update_order}
\caption{Example for homogeneous demand at $p_T = 0$ and the importance of update order. The tables describe the model parameters and the costs if agent $i$ would buy (is buying) at node $k$. Current suppliers are marked in green, if they are optimal, otherwise they are marked in blue. Optimal decisions are marked in yellow or red, depending on whether they can/cannot be found by a local update, respectively. Here, all agents want to update their current purchases (marked in blue), but only agent $4$ can find its optimal supplier locally through its neighbor (marked in yellow). Of all possible (local) updates, this update has the smallest $K_i(k)$ and would be executed first, allowing other agents to follow and find their optimal supplier locally. As shown above, with homogeneous demand if the network is not in equilibrium we will always find an optimal local update by sorting for the smallest $K_i(k)$.}
\end{figure}

\newpage
\section{Simulation}

With the algorithm introduced above we can solve the the individual optimization problems with only local updates for a given value of $p_T$ when the production costs per unit are non-increasing and the demand of all agents is identical. We now discuss how we exactly track the evolution of the system for the whole range of decreasing transaction costs $p_T$ (increasing transaction costs are handled analogously). 
First, we note that it is not necessary to calculate the distances $T_{ki}$ for all shortest (cheapest) paths, since these paths will be found automatically during the optimization.\\

We track for each agent $i$ its current supplier $i'$ and the distance $T_{i'i}$, starting from the trivial initial condition $i' = i$ for all agents at $p_T = \infty$ ($T_{ii} = 0$). For all (directed) links we then explicitly calculate the values $p_T^{ji}$ where the link would be active, i.e., when $i$ starts buying via $j$ from $j'$, where $j'$ is the supplier of $i$'s neighbor $j$. As in the algorithm above, we also explicitly consider returning to internal production as a special case. Comparing the cost of the current supplier $i'$ and the new supplier $j'$:

\begin{eqnarray}
	K_i(i') &=& D \cdot p_{i'}\left[S_{i'}^i + D\right] + D \cdot p_T \cdot T_{i'i},\label{eq:cost_ii*}\\
	K_i(j'\;\mathrm{via}\;j) 
	&=& D \cdot p_{j'}\left[S_{j'}^i + D\right] + D \cdot p_T \cdot \left(T_{j'j} + t_{ji}\right) \,, \label{eq:cost_ij}
\end{eqnarray}
we find that $i$ would prefer to buy from $j'$ when
\begin{eqnarray}
	K_i(j'\;\mathrm{via}\;j) &<& K_i(i') \nonumber\\
	p_T \cdot \underbrace{\left(t_{ji} + T_{j'j} - T_{i'i}\right)}_{\Delta T} &<& \underbrace{ p_{i'}\left[S_{i'}^i + D\right] - p_{j'}\left[S_{j'}^i + D\right] }_{\Delta p} \,. \label{eq:switch_supp}
\end{eqnarray}
With this inequality we can easily calculate  $p_T^{ji}$. We create the sorted list $\mathcal{U} = \left\{ \left(p_T^{ji}, K_i(j'), e_{ji}\right), \dots\right\}$, ordering all links in descending order in $p_T^{ji}$. If multiple links have the same $p_T^{ji}$, e.g., for simultaneous updates during cascades, we use the new final costs $K_i(j')$ as a secondary criterion for ordering the updates (increasing in $K_i(j')$), as discussed above. We then have to distinguish three different cases to determine $p_T^{ji}$:
\begin{itemize}
	\item $\Delta T > 0$\\
	in this case we simply have the condition $p_T < \frac{\Delta p}{\Delta T}$. We thus set $p_T^{ji} = \frac{\Delta p}{\Delta T}$ if the condition is not yet fulfilled and $p_T^{ji}=p_T$ otherwise. To calculate $K_i(j')$ for the secondary ordering we use the transaction costs $p_T^{ji}$, as this is the value when the ordering will be relevant. If no other changes occur and this link is not updated again, this secondary ordering is then already accurate when $p_T$ decreases far enough (and will be updated in between otherwise).
	
	\item $\Delta T = 0$\\
	in this case the switch does not depend on the transaction costs at all but the condition is $\Delta p > 0$. If this condition is fulfilled, we set $p_T^{ji} = p_T$, such that the switch will be applied immediately. As above, the secondary ordering is given by the final cost. Otherwise we set $p_T^{ji} = -1$ and the switch is never executed.
	
	\item $\Delta T < 0$\\
	in this case the condition becomes $p_T > -\frac{\Delta p}{\left|\Delta T\right|}$ (note the change in the sign of the inequality). In case the condition is true, the change is applied instantly and we set $p_T^{ji} = p_T$, otherwise we set $p_T^{ji} = -1$. This case will mostly be relevant when considering increasing transaction cost.
	
\end{itemize}

We then use this list $\mathcal{U}$ to iteratively execute single updates. The link with the largest $p_T^{ji}$ (and smallest final total cost $K_i(j')$, if applicable) will become active first, meaning $i$ will start buying via $j$ from $j$'s supplier $j'$. The algorithm then follows an event-based methodology:\\

Starting from $p_T = \infty$ and only internal production $i^* = i$ for all nodes $i \in \left\{1,\dots,N\right\}$, we first decrease the value of $p_T$ to the value of the next update $p_T \leftarrow \max_\mathcal{U} p_t^{ji}$. Then we update $i' \leftarrow j'$ and $T_{i'i} \leftarrow T_{j'j} + t_{ji}$. We also adjust the production of the involved suppliers $i'$ and $j'$ accordingly. Due to the economies of scale in the production costs per unit these changes will affect the values $p_T^{kl}$ of all links adjacent to agents currently buying from $i'$ or $j'$, for which we recalculate $p_T^{kl}$.\\

We then repeat the updating process working through all links with $p_T^{ji} = p_T$ and keep updating the purchases until no further changes occur and the next entry in the update list occurs for $p_T^{ji} < p_T$. Finally, we repeat this whole process by further decreasing $p_T$. In contrast to the algorithm for a single update step given above, we will never have $\mathcal{U} = \varnothing$. Instead, we terminate once all $p_T^{ji} \le 0$ and the transaction cost per unit cannot be reduced further.\\

The final state of the network will then be given at $p_T = 0$, where generally all agents have the same supplier and form one large cluster. Note that this updating scheme models a very slow decrease of the transaction costs. In principle it is also possible to include sudden jumps in $p_T$. Due to the individual optimization, this generally leads to a different evolution of the clusters (see Sec.~VI for examples).\\

We then reverse the process and consider increasing transaction cost in the same way by adjusting the conditions for $p_T^{ji}$ and reversing the ordering (executing updates with the smallest $p_T^{ji}$ first). In this case we terminate when $p_T^{ji} = \infty$ for all possible updates. Overall, this allows us to exactly track the evolution of all clusters in the network for both increasing as well as decreasing transaction costs per unit.\\

\newpage

\section{Network models}

\subsection{Spatially embedded random networks}

We create a spatially embedded planar random network with $N$ nodes by distributing $N$ points uniformly at random in the unit square. The links of the network are then created using the Delaunay-triangulation of these points [Fig.~\ref{fig:spatial_network}(a)]. Each link $e$ is assigned a distance $t_e$ equal to the euclidean distance between its two endpoints. We further assume an identical demand $D_i = 1/N$ for all nodes.

In order to reasonably study different system sizes (see also Sec.~IV for a detailed discussion) we define the parameters $b_i$ of the cost functions as values of a smooth random function on the unit square. The function is generated as the real part of a random periodic function calculated by spectral synthesis of $(2^10)^2$ Fourier-modes for frequency components $\omega_x,\omega_y \in \left\{0,1,\dots, 2^10-1\right\}$. For each Fourier mode we choose a Gaussian distributed amplitude with mean $0$ and variance $\left(\omega_x^2 + \omega_y^2\right)^{-2}$ and random phase \cite{saupe88_fractal_algorithms}. Inverse Fourier transformation then gives the function $b(x,y)$ at discrete points which we then interpolate at the positions of the nodes. Finally we shift and scale the values such that $b_i \in \left[0,1\right]$ with $\min_i (b_i) = 0$ and $\max_i (b_i) = 1$.

While this technically allows negative production costs per unit when including the economies of scale, a constant shift $b_k \rightarrow b_k + \mathrm{const.}$ for all nodes $k \in \left\{1,\dots,N\right\}$ does not change the solution of the minimization problem. Therefore, for simplicity we only consider normalized $b_k \in \left[0,1\right]$ in all examples.\\

\begin{figure}[h]
\centering
\includegraphics[width=0.5\textwidth]{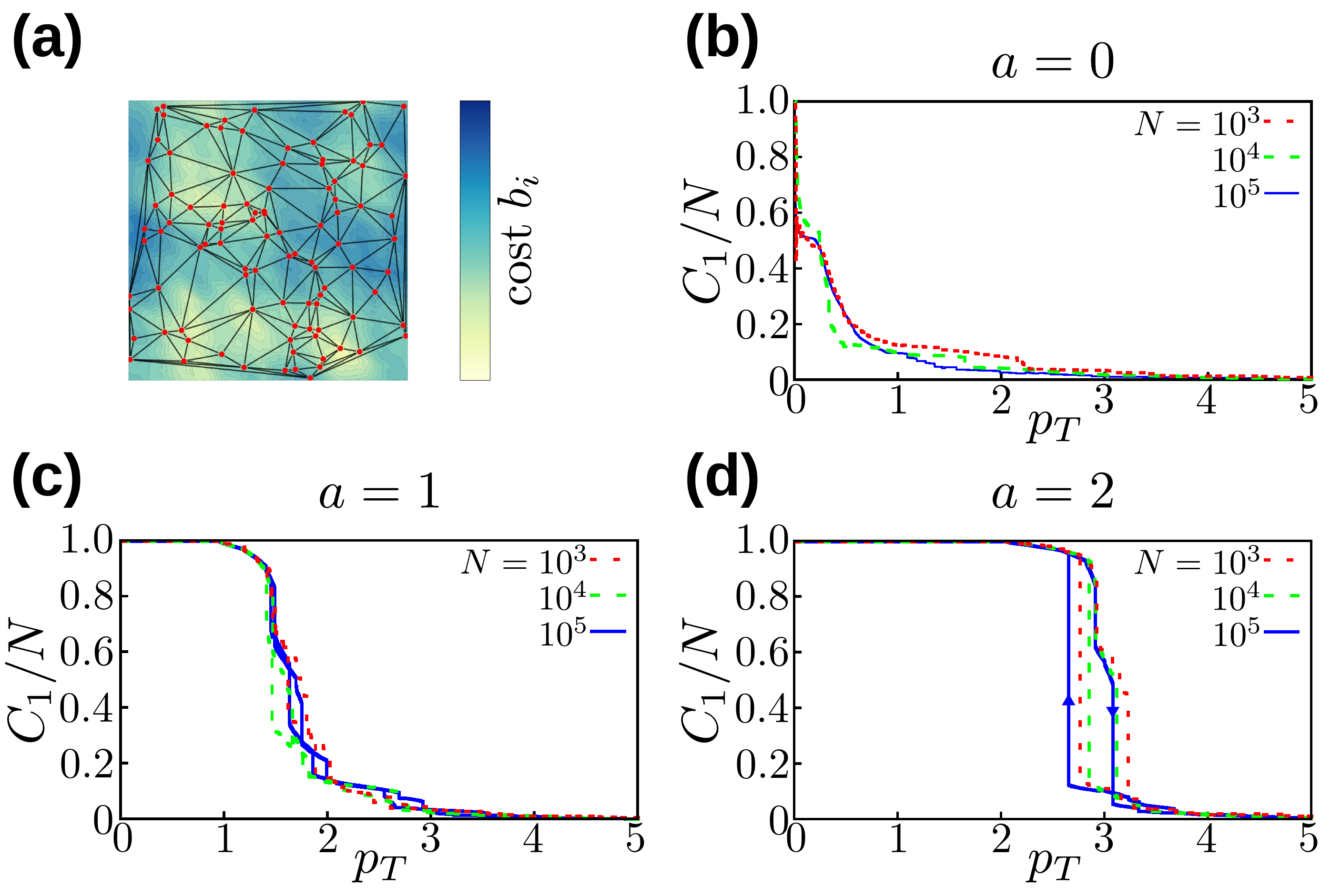}
\label{fig:spatial_network}
\caption{Spatially embedded planar random network. (a) A single realization of a network with $N = 100$ nodes together with the smooth function $b(x,y)$ determining the parameters $b_i$. (b-d) Evolution of the model for different values of the economies of scale $a$. The evolution is similar in for all system sizes and the limit of large networks $N \rightarrow \infty$ is well defined due to the identical scaling of distances and cost differences $\left|b_i - b_j\right|$ since $b(x,y)$ is a smooth function.}
\end{figure}

\subsection{Small world networks}

To illustrate the effects of the underlying network topology we consider Watts-Strogatz small-world networks \cite{watts98_smallworld}. This model allows generating networks that interpolate between regular and random topologies with a large or small network diameter, respectively: Start from a ring of $N$ nodes where each node is connected to its $k/2$ nearest neighbors on each side. Each link in the network is then rewired with probability $q_\mathrm{rew}$, i.e., disconnected on one end and connected to another node in the network, chosen uniformly at random. This procedure interpolates between regular networks with a large diameter for $q_\mathrm{rew} \rightarrow 0$ and completely random networks with a small diameter for $q_\mathrm{rew} \rightarrow 1$.

For the parameters of our model we assume identical transaction costs for all links $t_e = 1/N$ as well as identical demand $D_i = 1/N$ for all nodes. The cost function parameters are taken as $a = 1$ and uniformly randomly chosen $b_i \in \left[0,1\right]$ (see also Sec.~V).\\

\newpage

\section{Expansion vs. fine-graining}

Usually, in percolation models the thermodynamic limit of large systems $N \rightarrow \infty$ is relevant for studying the (phase-)transition to global connectedness. We will now motivate our choice to consider an increased system size $N$ as a fine-graining of the system.

Apart from the obvious inconsistencies with the models motivation arising for infinitely expanding systems, e.g., an infinite system size would also mean infinite demand, other problems arise due to the interplay between economies of scale, initial production costs per unit, and transaction costs. To illustrate this, we assume finite, bounded intervals for all parameters. Specifically we take the production costs per unit to be $p_k(S) = b_k - aS(k)$ with $b_k \in \left[b_\mathrm{min}, b_\mathrm{max}\right]$ and similarly, $D_i = D$ and $t_{ij} \in \left[t_\mathrm{min}, t_\mathrm{max}\right]$ with $t_\mathrm{min} > 0$. For simplicity, we consider the model without economies of scale in the limiting case $a=0$.\\

For a given value of the transaction cost $p_T$ we will now derive a (very loose) upper bound for the size of a cluster. Consider the longest shortest path $\Pi_{ki,\mathrm{max}}$ in one cluster between agent $i$ and its supplier $i^*$. The transaction cost along those paths cannot be higher than the difference between the production costs per unit of the two nodes:
\begin{eqnarray}
	b_{i^*} + p_T T_{i^*i} &\le& b_i \nonumber\\
	\Rightarrow \quad p_T T_{i^*i} &\le& b_i - b_{i^*} \le b_\mathrm{max} - b_\mathrm{min}\nonumber\,,
\end{eqnarray}
which means that there are at most $d_\mathrm{max}$ links involved in the path, where
\begin{equation}
	d_\mathrm{max} \le \frac{T_{i^*i}}{t_\mathrm{min}} \le \frac{b_\mathrm{max} - b_\mathrm{min}}{p_T t_\mathrm{min}} \,.
\end{equation}
If the underlying network has a maximum degree $k_\mathrm{max}$, we can then find an upper bound for number of nodes in the largest cluster as
\begin{equation}
	C_1(p_T) \le 1 + k_\mathrm{max} + k_\mathrm{max}^2 + \dots + k_\mathrm{max}^{d_\mathrm{max}} = \sum_{d=0}^{d_\mathrm{max}} k_\mathrm{max}^d < \infty
\end{equation}
for any $p_T > 0$. Thus for a given range of parameters as assumed above in any network with bounded degrees we find that in the limit of $N \rightarrow \infty$ the size of the largest cluster will be finite for any $p_T > 0$. Thus the fraction of nodes in the largest cluster will follow
\begin{align}
	C_1/N = \begin{cases}
					0 \quad \mathrm{for} \quad p_T > 0 \\
					1 \quad \mathrm{for} \quad p_T = 0 \,,
				\end{cases}
\end{align}
describing a trivial transition at $p_T = 0$. More importantly, this even holds in cases where, for example, $k_\mathrm{max} \sim \ln(N)$, such as simple random graphs. Similarly, for $a>0$ the infinite demand will (in most cases) cause an infinite hysteresis loop for the system in the limit $N \rightarrow \infty$.\\

While these are interesting results, the more reasonable assumption of fine-graining the system maintains the balance between economies of scale and transaction cost such that the transitions and all related phenomena can be observed in a finite, non-trivial parameter interval and results for different system sizes can be compared more easily. Of course this requires the correct scaling of transaction costs, demand and network structure when increasing the system size.\\

This balance holds for example for the spatially embedded planar random network. There, the differences $\left|b_i - b_j\right|$ scale in the same way as the distances with increasing system size since the function $b(x,y)$ is smooth. This means that ratios $\left|b_i - b_j\right| / t_{ij}$ that define when links become active (see Sec.~III above) are bounded from below for neighboring nodes and the transition to global connectedness will occur for $p_T > 0$. Similarly, keeping the total demand fixed with increasing network size means that the ratio (now including economies of scale) is bounded from above and the largest cluster will break down for $p_T < \infty$. This ensures that both the initial transition as well as the potential hysteresis loop can be observed on the same scale. However, for networks without an underlying geometry this can be more difficult and there is not necessarily a correct scaling for all states of the network, as we discuss below for small world networks.\\
\newpage

With a similar argument as above, we can derive a sufficient condition for a discontinuous transition in any network topology. As soon as a single purchase reduces the production costs at a node $k$ so much that it becomes the optimal choice for all nodes to buy from this supplier a discontinuous transition occurs since all nodes will join this supplier in a large cascade of decisions. This is equivalent to the condition
\begin{equation}
	p_k(3D) + p_T \, T_{ki} < p_j(2D) + p_T \, T_{ji}
\end{equation}
for the cost function $p_k$ for all other nodes $i,j \in \left\{1,2, \dots ,N\right\}$ and $i \neq j$, meaning that node $i$ would prefer to buy from node $k$ compared to all other nodes $j$. In the limit of large networks $N \rightarrow \infty$ with $D = 1/N \rightarrow 0$, we find to first order in $D$
\begin{eqnarray}
	p_k(0) + 3D \, \frac{\mathrm{d} p_k(0)}{\mathrm{d} S_k} + p_T T_{ki} &<& p_j(0) + 2D \, \frac{\mathrm{d} p_j(0)}{\mathrm{d} S_j} + p_T T_{ji} \,,\nonumber\\
	p_\mathrm{max}(0) + 3D \, \frac{\mathrm{d} p_k(0)}{\mathrm{d} S_k} + p_T T_\mathrm{max} &<& p_\mathrm{min}(0) + 2D \, \frac{\mathrm{d} p_j(0)}{\mathrm{d} S_j} 
\end{eqnarray}
by substituting the most extreme values possible (maximum cost and distances on the left side and minimal cost and distances on the right). We then arrive at a (very loose) upper bound for the economies of scale by setting $\mathrm{d} p_k(0) / \mathrm{d} S_k = \mathrm{d} p_j(0) / \mathrm{d} S_j = \left(\mathrm{d} p / \mathrm{d} S\right)_\mathrm{min}$,
\begin{equation}
	\left(\mathrm{d} p / \mathrm{d} S\right)_\mathrm{min} < - \frac{\Delta p + p_T T_\mathrm{max}}{D} \,,
\end{equation}
where $\Delta p = p_\mathrm{max}(0) - p_\mathrm{min}(0)$ and $p_T$ is at most the value of the transaction cost where the first transition takes place. If the economies of scale are stronger than given in this condition, the first external purchase necessarily triggers a cascade of decisions and thus leads to a discontinuous decision.\\
\newpage

\section{Small world networks}

In the main manuscript we studied our optimization model on an underlying small world network \cite{1998_WATTS_small_world_model}: Start from a ring of $N$ nodes where each node is connected to its $k/2$ nearest neighbors in each direction. Each link in the network is then rewired with probability $q_\mathrm{rew}$, i.e., disconnected on one end and connected to another node in the network chosen uniformly at random.\\

In the main manuscript we showed for systems of size $N = 10^4$ how the topology changes the abruptness of the transition and the number of large clusters emerging. For completeness, Fig.~\ref{fig:small_world_a} shows the (inverse) size of the largest gap for different strengths of the economies of scale $a$ for the same system parameters ($N=10^4$, $t_e = 1/N$, $D_i = 1/N$, $b_i$ uniformly random in $\left[0,1\right]$). A network with large diameter ($q_\mathrm{rew} \rightarrow 0$) allows different clusters to emerge in remote parts of the network since the distances are large. Hence the jumps of the size of the largest cluster, denoted by $\Delta C_1$, are typically smaller ($1 - \left<\mathrm{max} \Delta C_1 / N\right>$ is large). However, very strong economies of scale ($a > 1$) can compensate for the large distances. Small fluctuations are amplified and one cluster will grow in a single cascade even in the network with large diameter. In contrast, for small diameters ($q_\mathrm{rew} = 0.01$) the jump $\Delta C_1$ grows already for moderate values of $a$.\\

\begin{figure}[h]
\centering
\includegraphics[width=0.5\textwidth]{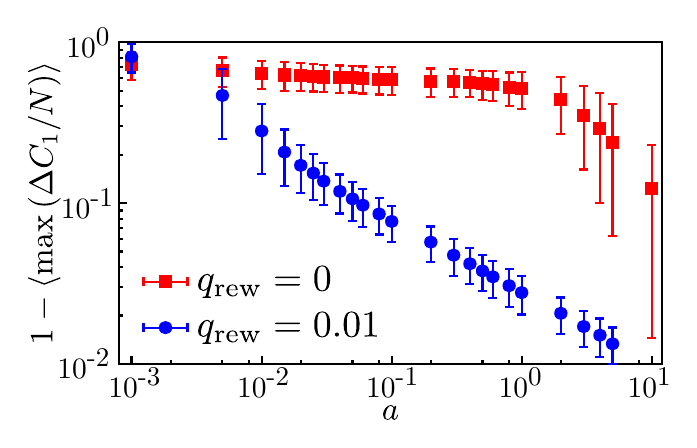}
\caption{Strong economies of scale lead to large cascades, independent of network topology. The inverse size of the largest gap in the size of the largest cluster (1 - largest gap) depending on the economies of scale $a$. For very weak economies of scale the cascades are small (the inverse is large), independent of the network topology a single agent does not have a large enough effect on the production costs. As the economies of scale become stronger, the average size of the largest cascade increases quickly in small world networks; due to the shortcuts a large cluster can affect all parts of the network. For regular networks with a large diameter ($q_\mathrm{rew} \rightarrow 0$) cascades stay small as transaction costs are still comparatively large. Unsurprisingly, for very strong economies of scale a single agent will change the production cost enough that the largest cluster grows discontinuously in a single cascade, even if the distances in the network are large.}
\label{fig:small_world_a}
\end{figure}

Additionally, we want to highlight some difficulties with this network structure in our model. As discussed above, to avoid trivial behavior in the limit $N \rightarrow \infty$, we need to scale all parameters, such that an increased system size corresponds to a fine graining of the system. However, in this case choosing a scaling of all variables to effectively fine-grain the system with increasing system size is not possible across different rewiring probabilities: the required scaling for a regular network with diameter $d \sim N$ will not be appropriate for a random network with diameter $d \sim \log(N)$.\\

Similarly, the scaling of $k$ and the choice of the $b_i$ over different system sizes is not uniquely defined. The behavior of the model in the limit of large systems $N \rightarrow \infty$ thus depends on the choice of the other parameters of the system. For example, if we choose $b_i$ uniformly at random without correlations between neighboring nodes and keep the cost per link constant, the limiting behavior will be a single cascade for very small $p_T$ where a microscopic cluster grows discontinuously to span the whole network (similar to the discussion above).

\newpage

\section{Analytically solvable example}

In addition to the complete graph discussed in the main manuscript we present another analytically solvable example. We consider a square lattice with a well defined single source and the same affine linear production costs per unit as in the main manuscript. The qualitative results are similar to those obtained for the complete graph, however, the actual evolution is more similar to the results obtained for the spatially embedded planar random network.\\

We place $N = L^2$ nodes as a square grid in the unit square. We index the nodes as node $\left(x,y\right)$ at position $\left(x/L,y/L\right)$, $x,y \in \left\{1,2,\dots ,L\right\}$. As in all other cases we have a homogeneous demand distribution $D_{\left(x,y\right)} = 1/N$. The production cost parameter $b_{\left(x,y\right)}$ at each node is given as $b_{\left(x,y\right)} = 1$ for all nodes, except for the node in the middle of the grid, where $b_{\left(\left\lfloor L/2 \right\rfloor, \left\lfloor L/2 \right\rfloor \right)} = 0$ (see Fig.~\ref{fig:single_source_square_lat_model}). Increasing the number of nodes in this system is equivalent to a fine-graining of the unit square.\\

\begin{figure}[h]
\centering
\includegraphics[width=0.5\textwidth]{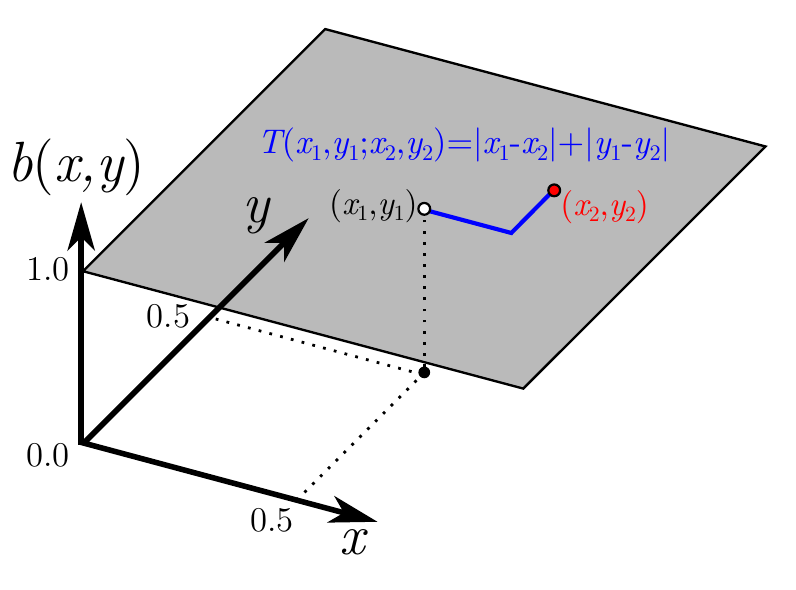}
\caption{Continuum illustration of the single source square lattice model. We have $b_{\left(x,y\right)} = 1$ for all points except for the center point in the lattice at $(1/2,1/2)$. The distance between two points $T_{x_1,y_1;x_2,y_2}$ is given by the 1-norm, corresponding to the shortest path in the square lattice. This continuum form of the model can easily be solved analytically while the corresponding lattice model can be simulated efficiently.}
\label{fig:single_source_square_lat_model}
\end{figure}

Due to the symmetry and the single source only one cluster will emerge and we can easily calculate its size. For simplicity we consider the continuous version of the problem illustrated in Fig.~\ref{fig:single_source_square_lat_model}, i.e., the limit $N \rightarrow \infty$. Due to the symmetry of the problem, all agents at the same distance $\mathrm{dist}\left[ (x,y); (1/2,1/2) \right] = \left|x - 1/2\right| + \left|y - 1/2\right|$ to the central supplier will behave identically. Note that distances are given as the $1$-norm due to the underlying square grid network structure. Thus, we can derive a self-consistency equation, describing the size of the cluster in terms of the maximum distance $d$ of an agent still willing to buy from the central supplier. For $d < 1/2$ the shape of a cluster is a diamond (due to the $1$-norm), but for $d > 1/2$ the corners are cut off, as we only consider points in the unit square. For given economies of scale $a$ and transaction costs per unit $p_T$ we determine $d$ as the value where the cost of internal production is equal to the cost of buying from the central supplier:
\begin{eqnarray}
	-a \left(2 d^2\right) + p_T d &=& 1 \quad\quad d \le 1/2 \\
	-a \left[ 2d^2 - 4\left(d-1/2\right)^2 \right] + p_T d &=& 1 \quad\quad d > 1/2 \;,
\end{eqnarray}
where the first term describes the economies of scale due to the size of the cluster and the second term is the transaction cost.\\

Solving these equations for $d$ and keeping in mind that $0 \le d(p_T) \le 1$ yields
\begin{eqnarray}
	d(p_T) = 	\begin{cases}
						\frac{p_T \pm \sqrt{p_T^2-8a}}{4a}	\quad&\mathrm{for}\quad d \le 1/2 \\
						\frac{4a - p_T \pm \sqrt{8a^2 + 8a - 8ap_T + p_T^2}}{4a} \quad&\mathrm{for}\quad d > 1/2  \,.\\
				\end{cases}
\end{eqnarray}
This equation has no, one or two possible solutions for $d$ for each $p_T$, however, each expression is valid only for some values of $d$. Therefore, as long as the bifurcation in the $d \le 1/2$ (lower branch) solution happens for $d > 1/2$, only one valid solution exists. We calculate the position of the bifurcation as $p_T^{\mathrm{crit,1}} = \sqrt{8a}$ and the corresponding distance $d_\mathrm{crit,1} = 1/\sqrt{2a}$. Thus, as long as $a < 2$ there is only one continuous solution for $d < 1/2$. We find a similar result for the solution of the $d > 1/2$ branch. Together, for $a < 2$ the two solutions join continuously and describe a single, continuous solution for $d(p_T$). As the size of the largest cluster is a continuous function of $d$, these results apply equally to $C_1$. Together, we find that the evolution of the largest cluster is continuous for $a < 2$, the slope of $C_1$ becomes infinite for $a = 2$ [see Fig.~\ref{fig:single_source_square_lat_simulation}(a)-(c)].

For $a > 2$, however, the first solution for $d < 1/2$ disappears at $p_T^{\mathrm{crit,1}} = \sqrt{8a}$ with $d_\mathrm{crit,1} = 1/\sqrt{2a} < 1/2$. There is no solution for smaller values of $p_T < p_T^{\mathrm{crit,1}}$ and two valid solutions for larger values. Similarly, the solution for the $d > 1/2$ branch does not exist for $p_T > p_T^{\mathrm{crit,2}} = 4a - \sqrt{8a(a-1)}$ with $d_\mathrm{crit,2} > 1/2$, but has two valid solutions for smaller values of $p_T$. The two solutions are still joined continuously, but now by an unstable branch: a slightly smaller cluster will rapidly shrink to a small stable size, while a slightly larger cluster will grow to the larger stable size. The transition for decreasing/increasing transaction costs becomes discontinuous and a hysteresis loop emerges in between [see Fig.~\ref{fig:single_source_square_lat_simulation}(d)].\\

We illustrate these analytical results in Fig.~\ref{fig:single_source_square_lat_simulation} together with simulations for $N = 25600$ for different values of $a$, showing good agreement with our analytical calculations. Contrary to the complete graph example in the main manuscript, the cluster grows by a finite amount already before the discontinuous transition takes place (similarly for increasing transaction costs). This is explained by the fact that close to the source in the middle of the grid, the difference in production cost per unit is constant while the transaction costs disappear for agents that are arbitrarily close to the center. Since the growth of the cluster is driven by these differences, the cluster grows as soon as $p_T < \infty$ whereas distances between two nodes are constant $t_{ij} = 1$ in the complete graph model and thus the cluster can only grow once $p_T$ decreases below a fixed, finite value.\\

\begin{figure}[h]
\centering
\includegraphics[width=0.45\textwidth]{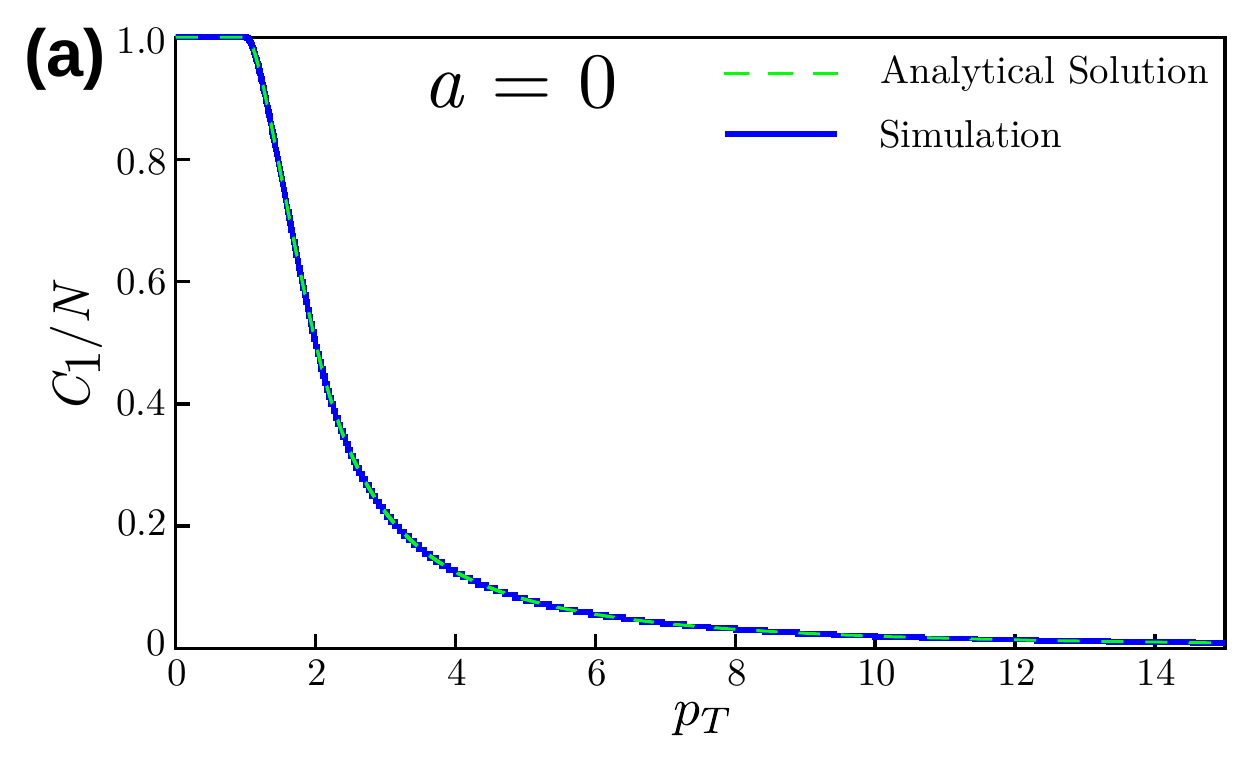}
\includegraphics[width=0.45\textwidth]{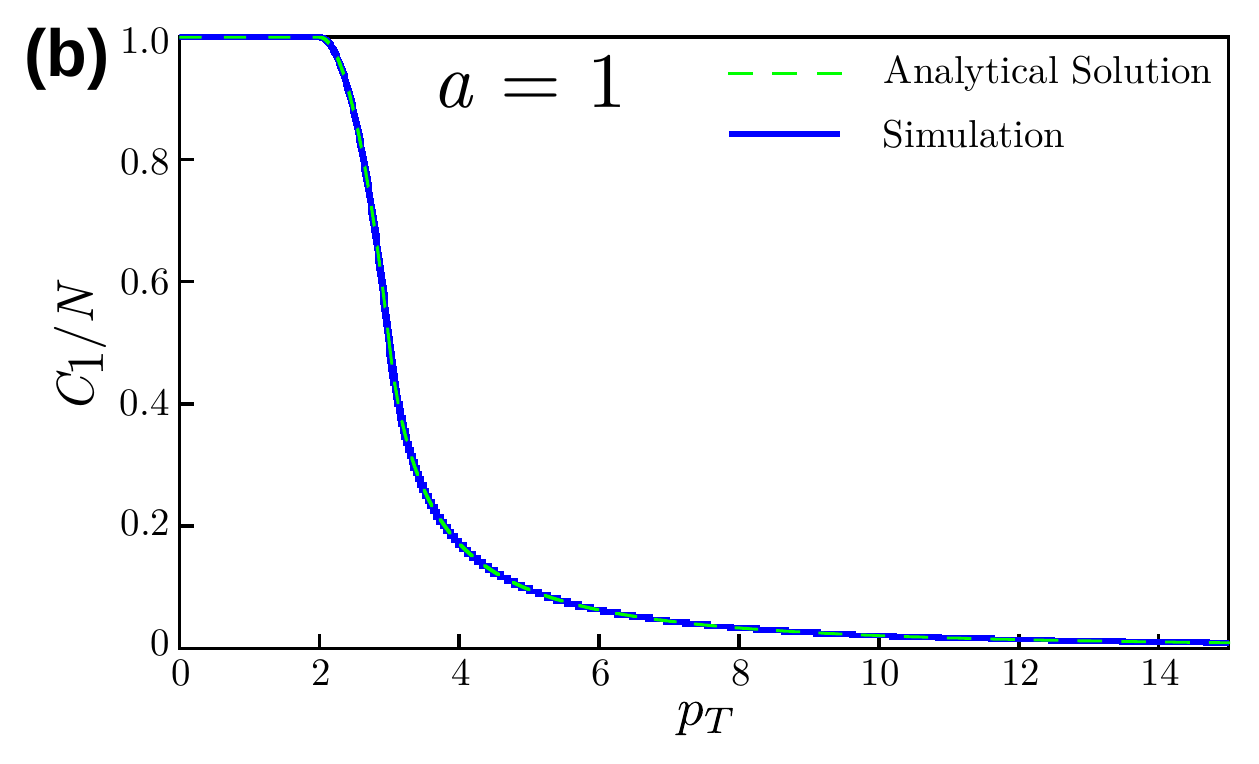}\\
\includegraphics[width=0.45\textwidth]{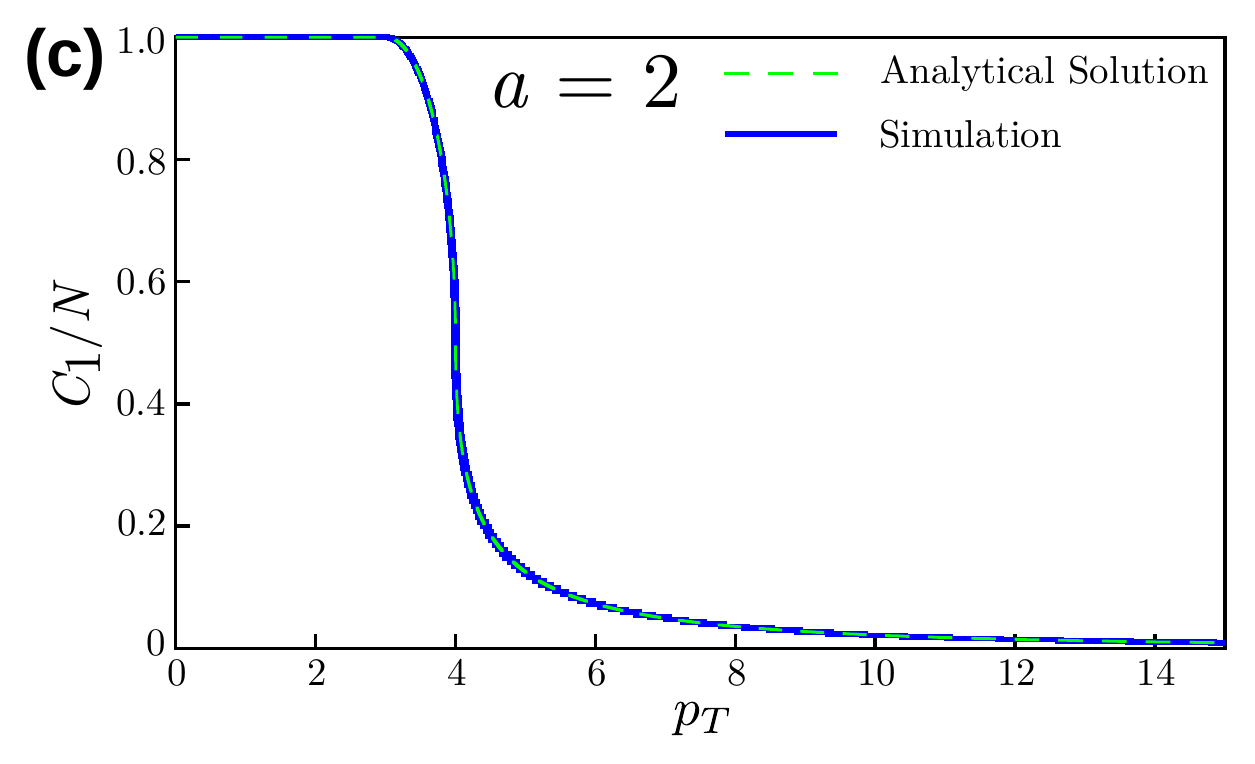}
\includegraphics[width=0.45\textwidth]{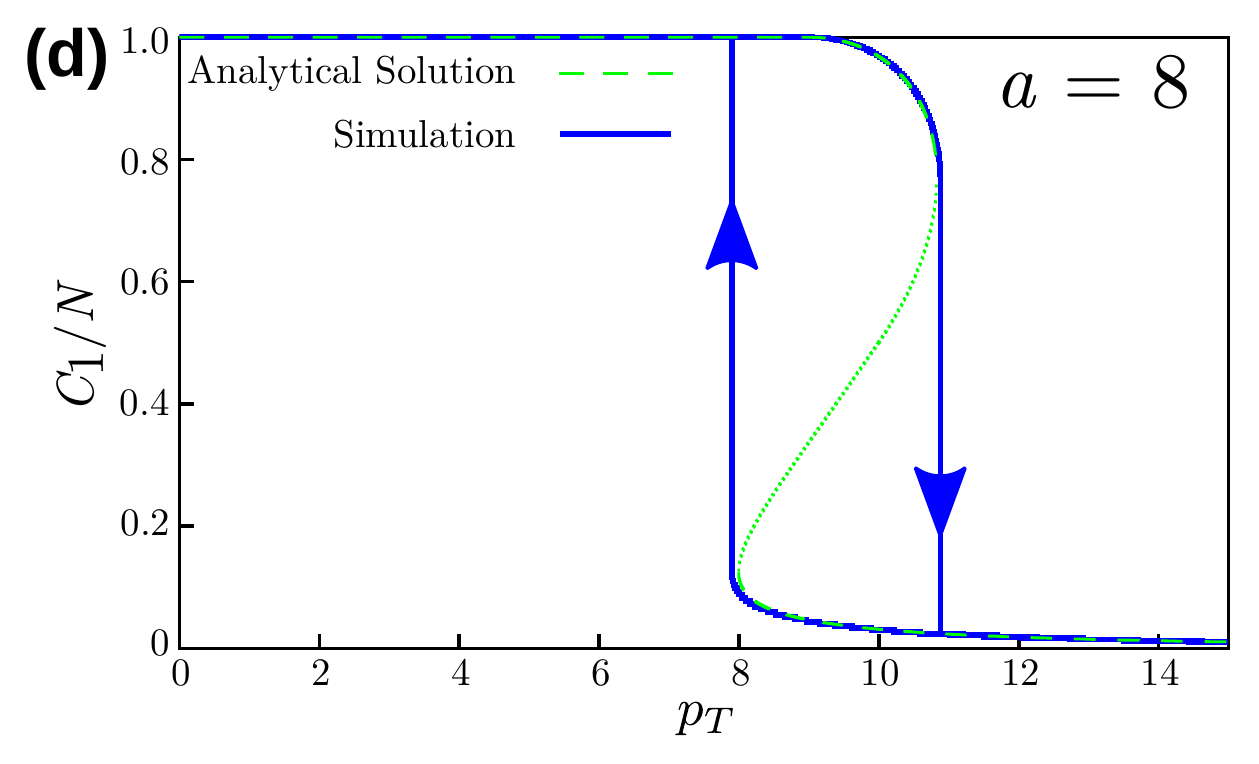}
\caption{Simulation of the single source square lattice model for $L = 160$ ($N = 25600$) and different $a$. (a),(b) The transition is continuous for small $a$, becoming steeper as $a$ increases. (c) At $a = 2$ the slope becomes infinite. (d) For larger $a$ the transition becomes discontinuous with a hysteresis loop. In the continuum self-consistency equation this structure emerges from a bifurcation of the possible solutions (dashed green lines) with an unstable branch in the middle (dotted line).}
\label{fig:single_source_square_lat_simulation}
\end{figure}

\clearpage

\section{World transport network}

In the following we describe the data used to create the world transport network. The nodes of the network describe individual countries or smaller regions of large countries (China, Russia, USA, Canada). The location of each node is given as the centroid of the respective country or the largest city in the corresponding region. For each of these countries we collected the total gross domestic product (GDP) and the total population as additional parameters. Further, we assign a harbor to each country with access to the sea, using the largest harbor in the country/region. The full world map with the location of all nodes is shown in Fig.~\ref{fig:worldmap_example}. See the Supplementary Data for all data and sources.\\

Transport links exist via land between countries sharing a border. The transaction cost of these links is simply given as the geodesic distance between the connected nodes. Transport links via sea exist between all countries/regions with a harbor. Here, the transaction cost is given as the sum of the distance of the corresponding node to the respective harbor for both countries and the length of the actual sea route, $t_{e_{ij}} = t_i^\mathrm{harbor} + t_j^\mathrm{harbor} + T_S t_\mathrm{sea}$, where $t_i^\mathrm{harbor}$ is the distance of the node corresponding to country $i$ to its largest harbor. These sea routes were determined as shortest paths across a triangulation of the worlds oceans, explicitly including important routes and channels (for example the Panama and Suez channel). In order to change the preferred mode of transport we introduce the factor $T_S$ describing the relative cost of sea travel compared to land travel. To determine the transaction cost for sea links we scale the length of the sea route with this factor, such that $T_S < 1$ effectively means that sea travel is  preferred, while for $T_S>1$ land travel is comparatively cheaper.\\

To determine the parameters for our model we used the population and GDP of each country: $b_i = b = \mathrm{const.}$, $a = 1$ and $D_i \sim 1.15 P_i + G_i$, where $P_i$ and $G_i$ are the population and GDP of the corresponding region relative to the world total. The demand is then scaled such that, as in the other examples, the total demand is $\sum_i D_i = 1$. Even though the demand is not identical for all nodes, we use the local percolation algorithm to solve the optimization problem. As discussed above, this means the algorithm does not necessarily make an optimal update each step and the solution using the exact algorithm may be different. For comparison, we show the exact solution (algorithm \ref{alg:alg1}) in Fig.~\ref{fig:worldmap_exact_result}. The qualitative and quantitative behavior is almost identical to the results shown in the main manuscript, differences to the exact solution are negligible.

\begin{figure}[h]
\centering
\includegraphics[width=0.9\textwidth]{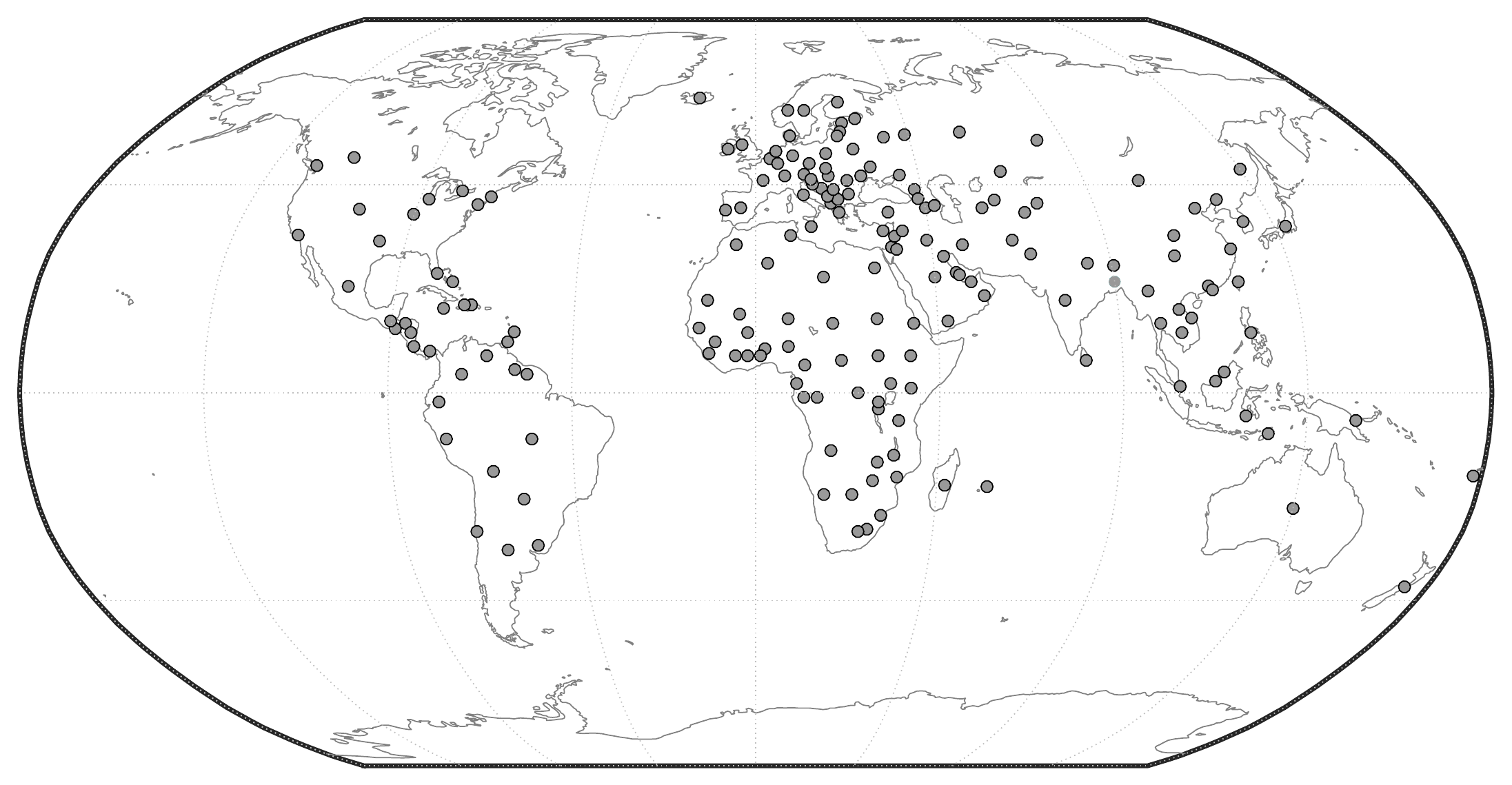}
\caption{Location of all nodes for the world transport network. Possible transport links (not shown) exist between countries sharing a border (via land) and between all countries with access to a harbor (via sea).}
\label{fig:worldmap_example}
\end{figure}

\newpage

Note, that this similarity is not a-priori obvious. Especially, differences will be larger in networks with a more heterogeneous demand distribution. Additionally, the structure of the network will affect the results: in this case, the differences are smaller when transport via sea is cheap. The shortest paths are the direct connections between nodes via sea such that the local percolation algorithm checks almost all nodes regardless. Results are more likely to differ for very sparse networks with a large diameter.\\

We also illustrate here the impact of large changes in the transaction costs. In the main manuscript and all previous discussions we assumed a slow, gradual decrease/increase of the transaction costs, but the algorithms also allow to simulate arbitrary changes. In Fig.~\ref{fig:worldmap_instant_decrease} we compare an instant decrease from very large transaction costs to a given value with the results obtained for a gradual decrease. Different suppliers emerge in both cases. More importantly, in the case of instantaneous reduction of transaction cost, the center of the emerging cluster changes, depending on the final value of the transaction costs, in a different way compared to the results for the gradual decrease. This further illustrates the impact of history on the evolution of the model.

\begin{figure}[h]
\centering
\includegraphics[width=0.9\textwidth]{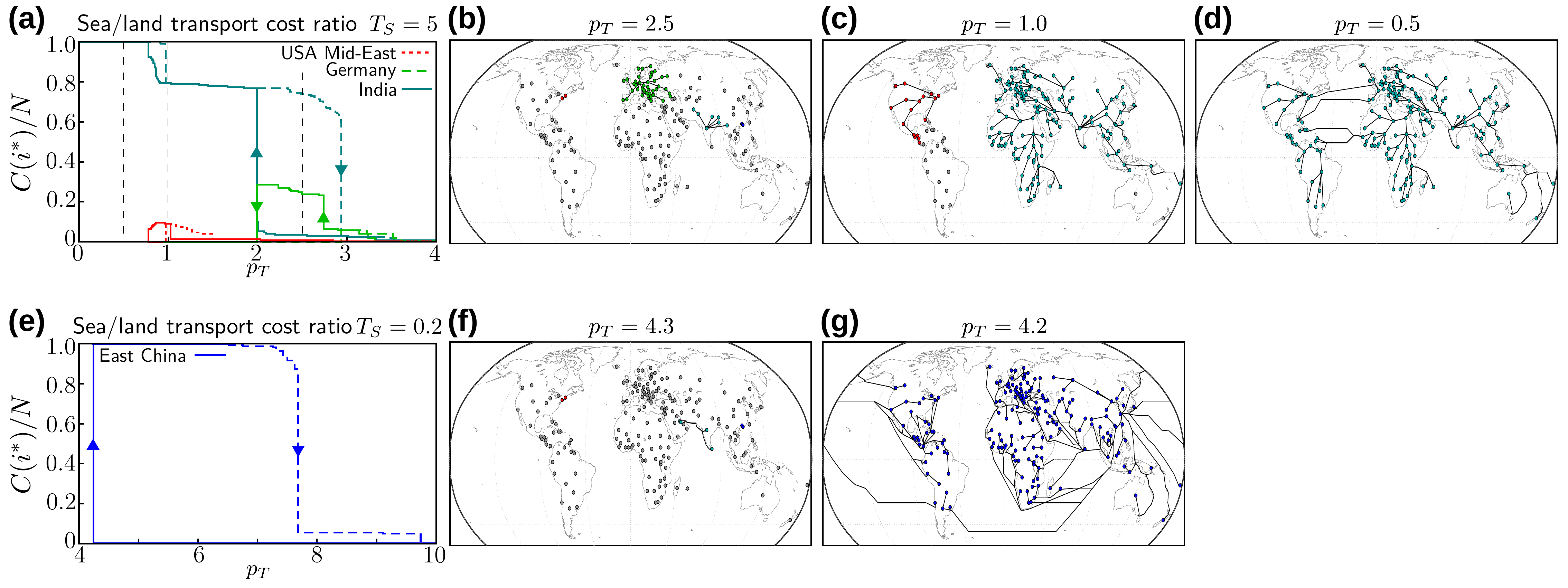}
\caption{Exact solution of the optimization model in a global transportation network. (a) Evolution of the connected components of individual suppliers for large sea travel costs ($T_S=5$). (b)-(d) Evolution of the individual clusters for different values of the transaction cost parameter $p_T$, land routes are preferred to transaction via sea.
(e) Evolution of the connected components of individual suppliers for small sea travel costs ($T_S=0.2$). (f),(g) State of the network immediately before and after the cascade.
The results are almost identical to those obtained with the local percolation algorithm (see main manuscript). The only difference is a small shift of $p_T$ for some of the transitions.}
\label{fig:worldmap_exact_result}
\end{figure}

\begin{figure}[h]
\centering
\includegraphics[width=0.9\textwidth]{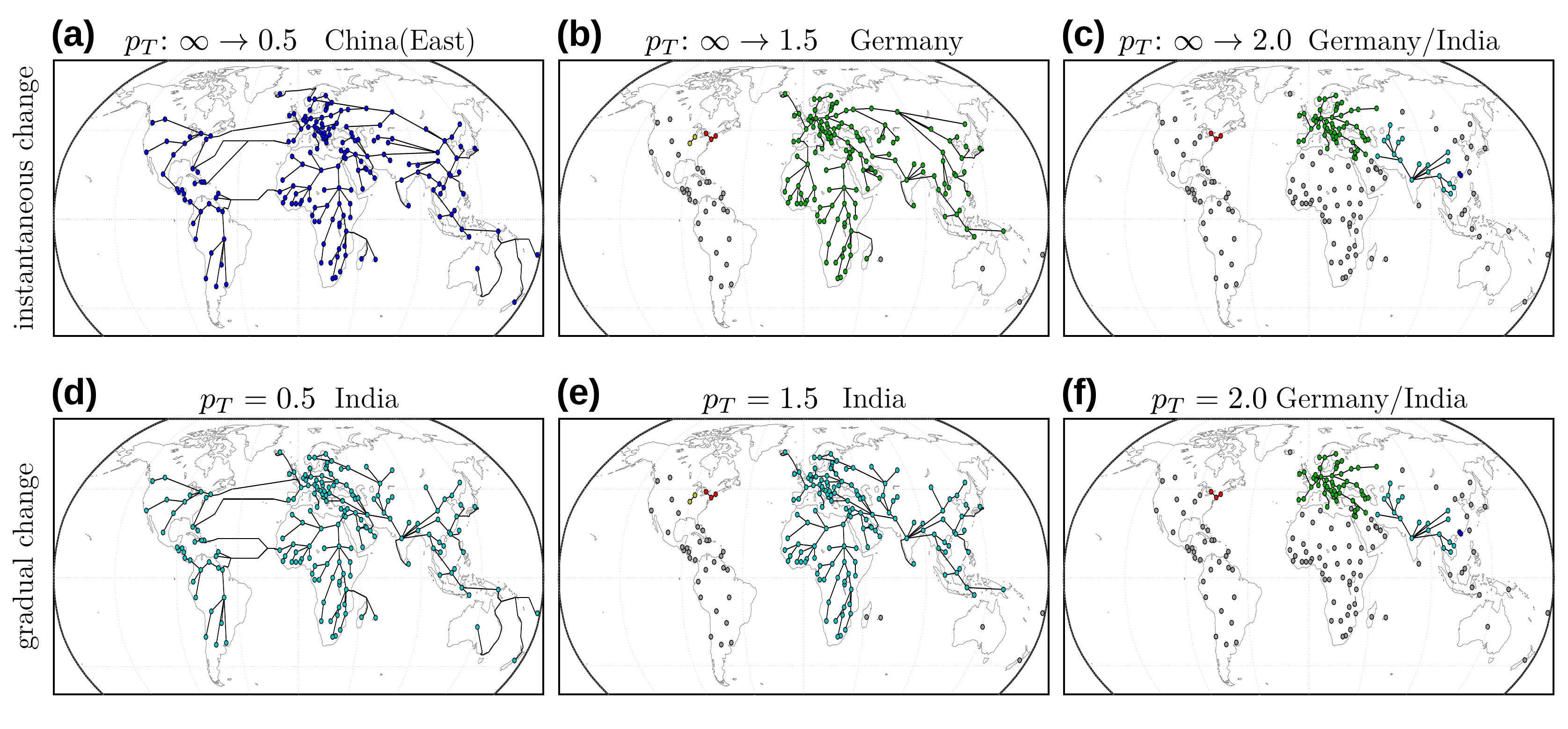}
\caption{Comparison of instantaneous and gradual decrease of the transaction cost. (a)-(c) Resulting network, when the transaction cost $p_T$ are reduced instantaneously to a given value (parameters are identical to other simulations, large sea travel costs $T_S = 5$). (d)-(f) Resulting network, when the transaction cost $p_T$ are reduced gradually. Depending on the final value of the transaction cost, different nodes emerge as suppliers of large clusters. Interestingly, if transaction costs are decreased far enough, the resulting network differs from the one obtained by gradual reduction of $p_T$ [panels (d-f)]. If only a few changes occur in the network [large transaction costs, panels (c) and (f)], there is no difference in the history and the resulting networks are identical. Otherwise, a gradual reduction leads to different updates before the final value of $p_T$ is obtained and consequently to generally different states. This further illustrates the multitude of stable states for any given value of the transaction costs.}
\label{fig:worldmap_instant_decrease}
\end{figure}

\end{document}